\documentclass[a4paper,UKenglish]{lipics}

\pdfoutput=1

\usepackage{amssymb}
\usepackage{latexsym}
\usepackage{amssymb}
\usepackage{amsmath}
\usepackage{mathtools}
\usepackage{cite}
\usepackage{textcomp}
\usepackage{hyperref}
\usepackage{url}
\usepackage[colorinlistoftodos]{todonotes}

\DeclareMathOperator{\landau}{\mathcal{O}}

\DeclareMathOperator{\npclass}{\mathsf{NP}}
\DeclareMathOperator{\ntime}{\mathsf{NTIME}}

\DeclareMathOperator{\pclass}{\mathsf{P}}

\DeclareMathOperator{\pspaceclass}{\mathsf{PSPACE}}

\DeclareMathOperator{\var}{\mathsf{var}}

\DeclareMathOperator{\rep}{\textsc{REP}}
\DeclareMathOperator{\3par}{\textsc{3-PAR}} 

\newcommand{\ta}{\ensuremath{\mathtt{a}}}
\newcommand{\tb}{\ensuremath{\mathtt{b}}}
\newcommand{\tc}{\ensuremath{\mathtt{c}}}

\theoremstyle{plain}

\newtheorem{proposition}{Proposition}
\newtheorem{claim}{Claim}
\makeatletter
\@addtoreset{claim}{theorem}
\newtheorem{problem}{Problem}

\usepackage{microtype}

\title{The Hardness of Solving Simple Word Equations}
\titlerunning{Simple Word Equations} 

\author[1]{Joel D. Day}
\author[1]{Florin Manea}
\author[1]{Dirk Nowotka}
\affil[1]{Kiel University, Department of Computer Science, D-24098, Kiel, Germany\\
  \texttt{\{jda,flm,dn\}@informatik.uni-kiel.de}}
\authorrunning{J.\,D. Day, F. Manea, and D. Nowotka} 

\Copyright{Joel\,D. Day, Florin Manea, and Dirk Nowotka}


\begin{document}

\maketitle

\begin{abstract}
We investigate the class of regular-ordered word equations. In such equations, each variable occurs at most once in each side and the order of the variables occurring in both sides is the preserved (the variables can be, however, separated by potentially distinct constant factors). Surprisingly, we obtain that  solving such simple equations, even when the sides contain exactly the same variables, is~$\npclass$-hard. By considerations regarding the combinatorial structure of the minimal solutions of the more general quadratic equations we obtain that the satisfiability problem for regular-ordered equations is in~$\npclass$. Finally, we also show that a related class of simple word equations, that generalises one-variable equations, is in $\pclass$.
\end{abstract}
\section{Introduction}\label{sec:intro}

A {\em word equation} is an equality $\alpha = \beta$, where $\alpha$ and $\beta$ are words over an alphabet $\Sigma \cup X$ (called the left, respectively, right side of the equation); $\Sigma=\{\ta,\tb,\tc,\ldots\}$ is the alphabet of \emph{constants} and $X = \{x_1, x_2, x_3, \ldots\}$ is the alphabet set of \emph{variables}. A \emph{solution} to the equation $\alpha = \beta$ is a morphism $h : (\Sigma \cup X)^* \to \Sigma^*$ that acts as the identity on $\Sigma$ and satisfies $h(\alpha) = h(\beta)$. For instance, $\alpha =x_1 \ta \tb x_2$ and $\beta = \ta x_1 x_2 \tb$ define the equation $x_1 \ta \tb x_2 = \ta x_1 x_2 \tb$,
whose solutions are the morphisms $h$ with $h(x_1) = \ta^k$, for $k\geq 0$, and $h(x_2) = \tb^\ell$, for $\ell \geq 0$. 

The study of word equations (or the existential theory of equations over free monoids) is an important topic found at the intersection of algebra and computer science, with significant connections to, e.g., combinatorial group or monoid theory~\cite{Lyndon77,Lyndon60,DiekertW16}, unification~\cite{sch:wor,Jaffar90,Jez14}), and, more recently, data base theory \cite{FreydenbergerH16,Freydenberger17}. 
The problem of deciding whether a given word equation $\alpha=\beta$ has a solution or not, known as the satisfiability problem, was shown to be decidable by Makanin~\cite{mak:the} (see Chapter~$12$ of \cite{lot:alg} for a survey). Later it was shown that the satisfiability problem is in $\pspaceclass$ by Plandowski~\cite{Pla2006}; a new proof of this result was obtained in~\cite{Jez2016b}, based on a new simple technique called recompression. However, it is conjectured that  the satisfiability problem is in $\npclass$; this would match the known lower bounds: the satisfiability of word equations is $\npclass$-hard, as it follows immediately from, e.g., \cite{ehr:fin}. This hardness result holds in fact for much simpler classes of word equations, like the quadratic equations (where the number of occurrences of each variable in $\alpha\beta$ is at most two), as shown in \cite{DieRob99}. There are also cases when the satisfiability problem is tractable. For instance, word equations with only one variable can be solved in linear time in the size of the equation, see~\cite{Jez2016}; equations with two variables can be solved in time $\landau(|\alpha\beta|^5)$, see~\cite{DabPla2004}.

In general, the $\npclass$-hardness of the satisfiability problem for classes of word equations was shown as following from the $\npclass$-completeness of the \emph{matching problem} for corresponding classes of patterns with variables. In the matching problem we essentially have to decide whether an equation $\alpha=\beta$, with $\alpha \in (\Sigma\cup X)^*$ and $\beta\in \Sigma^*$, has a solution; that is, only one side of the equation, called pattern, contains variables. The aforementioned results \cite{ehr:fin,DieRob99} show, in fact, that the matching problem is $\npclass$-complete for general $\alpha$, respectively when $\alpha$ is quadratic. Many more tractability and intractability results concerning the matching problem are known (see \cite{rei:patIaC, fer:patjour, FerSchVil2015}). In~\cite{FeMaMeSc14_stacs}, efficient algorithms were defined for, among others, patterns which are \emph{regular} (each variable has at most one occurrence), \emph{non-cross} (between any two occurrences of a variable, no other distinct variable occurs), or patterns with only a constant number of variables occurring more than once.

Naturally, for a class of patterns that can be matched efficiently, the hardness of the satisfiability problem for word equations with sides in the respective class is no longer immediate. A study of such word equations was initiated in \cite{ManSchNow16}, where the following results were obtained. Firstly, the satisfiability problem for word equations with non-cross sides (for short non-cross equations) remains $\npclass$-hard. In particular, solving non-cross equations $\alpha=\beta$ where each variable occurs at most three times, at most twice in $\alpha$ and exactly once in $\beta$, is $\npclass$-hard. Secondly, the satisfiability of one-repeated variable equations (where only one variable occurs more than once in $\alpha\beta$, but an arbitrary number of other variables occur only once) having at least one non-repeated variable on each side, was shown to be in~$\pclass$.

In this paper we mainly address the class of regular-ordered equations, whose sides are regular patterns and, moreover, the order of the variables occurring in both sides is the same. This seems to be one of the structurally simplest classes of equations whose number of variables is not bounded by a constant. Our central motivation in studying this kind of equations with a simple structure is that understanding their complexity and combinatorial properties may help us define a boundary between classes of word equations whose satisfiability is tractable and intractable, as well as to gain a better understanding of the core reasons why solving word equations is hard.
In the following, we overview our results, methods, and their connection to existing works from the literature.

\vspace*{0.1cm}

\noindent {\bf Lower bounds.} Our first result closes the main open problem from \cite{ManSchNow16}. Namely, we show that it is still $\npclass$-hard to solve regular (ordered) word equations. Note that in these word equations each variable occurs at most twice: once in every side. They are particular cases of both quadratic equations and non-cross equations, so the reductions showing the hardness of solving these more general equations do not carry over. To begin with, matching quadratic patterns is $\npclass$-hard, while matching regular patterns can be done in linear time. Showing the hardness of the matching problem for quadratic patterns in \cite{DieRob99} relied on a simple reduction from $3$-SAT: one occurrence of each variable of the word equation was used to simulate an assignment of a corresponding variable in the $3$-SAT formula, then the second occurrence was used to ensure that this assignment satisfies the formula. To facilitate this final part, the second occurrences of the variables were grouped together, so the equation constructed in this reduction was (clearly) not non-cross. Indeed, matching non-cross patterns can be done in polynomial time. So showing that solving non-cross equations is hard, in \cite{ManSchNow16}, required slightly different techniques. This time, the reduction was from an assignment problem in graphs. The (single) occurrences of the variables in one side of the equation were used to simulate an assignment in the graph, while the (two) occurrences of the variables from the other side were used for two reasons: to ensure that the previously mentioned assignment is correctly constructed and to ensure that it also satisfies the requirements of the problem. For the second part it was also useful to allow the variables to occur in one side in a different order than their order from the other side. 

As stated in \cite{ManSchNow16}, showing that the satisfiability problem for regular equations seems to require a totally different approach. Our hardness reduction relies on some novel ideas, and, unlike the aforementioned proofs, has a deep word-combinatorics core. As a first step, we define a reachability problem for a certain type of (regulated) string rewriting systems, and show it is $\npclass$-complete (in Lemma \ref{rep_is_hard}). This is achieved via a reduction from the strongly $\npclass$-complete problem \textsc{3-Partition}~\cite{gar:com}. Then we show that this reachability problem can be reduced to the satisfiability of regular-ordered word equations; in this reduction (described in the successive Lemmas \ref{lem:sol3}, \ref{lem:solutions}, and \ref{lem:sol2}), we essentially try to encode the applications of the rewriting rules of the system into the periods of the words assigned to the variables in a solution to the equation. In doing this, we are able to only use one occurrence of each variable per side, and moreover to even have the variables in the same order in both sides. This overcomes the two main restrictions of the previous proofs: the need of having two occurrences of some variables on one side and the need to have a different order of the variables in the two sides of the equation, or, respectively, to interleave the different occurrences of different variables.  

As a concluding remark, our reduction suggests the ability of this very simple class of equations to model other natural problems in rewriting, combinatorics on words, and even beyond. In this respect, our construction is also interesting from the point of view of the expressibility of word equations, such as studied in~\cite{kar:the}. 

\vspace*{0.1cm}

\noindent {\bf Upper bounds.} A consequence of the results in \cite{PlandowskiR98} is that the satisfiability problem for a certain class of word equations is in $\npclass$ if the length of the minimal solutions of such equations (where the length of the solution defined by a morphism $h$ is the image of the equation's sides under~$h$) are at most exponential. With this in mind, we show Lemma~\ref{lem:squares1}, which gives us an insight in the combinatorial structure of the minimal solutions of quadratic equations. Further, in Proposition~\ref{prop:linearsolutions}, we give a concise proof of the fact the image of any variable in a minimal solution to a regular-ordered equation is at most linear in the size of the equations (so the size of the minimal solutions is quadratic). It immediately follows that the satisfiability problem for regular-ordered equations is in $\npclass$. It is an open problem to show the same for arbitrary regular or quadratic equations, and hopefully the lemma we propose here might help in that direction. Also, it is worth noting that our polynomial upper bound on length of minimal solutions of regular-ordered equations is, in a sense, optimal. More precisely, non-cross equations $\alpha=\beta$ where the order of the variables is the same in both sides and each variable occurs exactly three times in $\alpha\beta$, but never only on one side, may already have exponentially long minimal solutions (see Proposition~\ref{prop:exp-length}). To this end, it seems even more surprising that it is $\npclass$-hard to solve equations with such a simple structure (regular-ordered), which, moreover, have quadratically short solutions.

In the rest of the paper we deal with a class of word equations whose satisfiability is tractable. To this end, we use again a reasoning on the structure of the minimal solutions of equations, similar to the above, to show that if we preserve the non-cross structure of the sides of the considered word equations, but allow only one variable to occur an arbitrary number of times, while all the others occur exactly once in both sides, we get a class of equations whose satisfiability problem is in $\pclass$. This problem is related to the one-repeated variable equations considered in \cite{ManSchNow16}; in this case, we restrict the equations to a non-cross structure of the sides, but drop the condition that at least one non-repeated variable should occur on each side. Moreover, this problem generalises the one-variable equations~\cite{Jez2016}, while preserving the tractability of their satisfiability problem. Last, but not least, this result shows that the pattern searching problem, in which, given a pattern $\alpha\in (\Sigma\cup\{x_1\})^*$ containing constants and exactly one variable $x_1$ (occurring several times)  and a text $\beta \in (\Sigma\cup\{x_1\})^*$ containing constants and the same single (repeated) variable, we check whether there exists an assignment of $x_1$ that makes $\alpha$ a factor of $\beta$, is tractable; indeed, this problem is the same as checking whether the word equation $x_2\alpha x_3=\beta, $ with $\alpha,\beta\in (\Sigma\cup\{x_1\})^*$, is satisfiable.

Due to space constraints, some proofs are given in the Appendix.

\section{Preliminaries}\label{sec:definitions}
Let $\Sigma$ be an alphabet. We denote by $\Sigma^*$ the set of all words over $\Sigma$; by $\varepsilon$ we denote the empty word. Let $|w|$ denote the length of a word $w$. For $1\leq i\leq j\leq |w|$ we denote by $w[i]$ the letter on the $i^{th}$ position of $w$ and $w[i..j]=w[i]w[i+1]\cdots w[j]$. A word $w$ is $p$-periodic for $p\in \mathbb{N}$ (and $p$ is called a period of~$w$) if $w[i]=w[i+p]$ for all $1\leq i\leq |w|-p$; the smallest period of a word is called its period. Let $w=xyz$ for some words $x,y,z\in \Sigma^*$, then $x$ is called prefix of $w$, $y$ is a factor of $w$, and $z$ is a suffix of $w$. Two words $w$ and $u$ are called conjugate if there exist non-empty words $x,y$ such that $w=xy$ and $u=yx$. 

Let $\Sigma=\{\ta,\tb,\tc,\ldots\}$ be an alphabet of \emph{constants} and let $X = \{x_1, x_2, x_3, \ldots\}$ be an alphabet of \emph{variables}. A word $\alpha \in (\Sigma \cup X)^*$ is usually called {\em pattern}. For a pattern $\alpha$ and a letter $z \in \Sigma \cup X$, let $|\alpha|_z$ denote the number of occurrences of $z$ in~$\alpha$; $\var(\alpha)$ denotes the set of variables from $X$ occurring in $\alpha$. A morphism $h : (\Sigma \cup X)^* \to \Sigma^*$ with $h(a) = a$ for every $a \in \Sigma$ is called a \emph{substitution}. We say that  $\alpha \in (\Sigma \cup X)^*$ is \emph{regular} if, for every $x \in \var(\alpha)$, we have $|\alpha|_{x} = 1$; e.\,g., $\ta x_1 \ta x_2 \tc x_3 x_4 \tb$ is regular. Note that $L(\alpha) = \{h(\alpha) \mid h \text{ is a substitution}\}$ (the pattern language of $\alpha$) is regular when $\alpha$ is regular, hence the name of such patterns. The pattern $\alpha$ is \emph{non-cross} if between any two occurrences of the same variable $x$ no other variable different from $x$ occurs, e.\,g., $\ta x_1 \tb \ta x_1 x_2 \ta x_2 x_2  \tb$ is non-cross, but $x_1 \tb x_2 x_2 \tb x_1$ is not. 

A \emph{word equation} is a tuple $(\alpha, \beta) \in (\Sigma \cup X)^+ \times (\Sigma \cup X)^+$; we usually denote such an equation by $\alpha = \beta$, where $\alpha$ is the left hand side (LHS, for short) and $\beta$ the right hand side (RHS) of the equation. A \emph{solution} to an equation $\alpha= \beta$ is a substitution $h$ with $h(\alpha) = h(\beta)$, and $h(\alpha)$ is called the \emph{solution word} (\emph{defined by $h$}); the length of a solution $h$ of the equation $\alpha=\beta$ is $|h(\alpha)|$. A solution of shortest length to an equation is also called minimal. 

A word equation is \emph{satisfiable} if it has a solution and the \emph{satisfiability problem} is to decide for a given word equation whether or not it is satisfiable.  
The satisfiability problem for general word equations is in $\ntime(n\log N)$, where $n$ is the length of the equation and $N$ the length of its minimal solution \cite{PlandowskiR98}. The~next~result~follows.

\begin{lemma}\label{lem:shortsolutions}
Let $\mathcal{E}$ be a class of word equations. Suppose there exists a polynomial $P$ such that such that for any equation in $\mathcal{E}$ its minimal solution, if it exists, has length at most $2^{P(n)}$ where $n$ is the length of the equation. Then the satisfiability problem for $\mathcal{E}$~is~in~$\npclass$.
\end{lemma}

A word equation $\alpha= \beta$ is regular or non-cross, if both $\alpha$ and $\beta$ are regular or both $\alpha$ and $\beta$ are non-cross, respectively; $\alpha=\beta$ is {\em quadratic} if each variable occurs at most twice in $\alpha\beta$. We call a regular or non-cross equation {\em ordered} if the order in which the variables occur in both sides of the equation is the same; that is, if $x$ and $y$ are variables occurring both in $\alpha$ and $\beta$, then all occurrences of $x$ occur before all occurrences of $y$ in $\alpha$ if and only if all occurrences of $x$ occur before all occurrences of $y$ in $\beta$. For instance $x_1x_1\ta x_2x_3 \tb =x_1\ta x_1x_2 \tb x_3$ is ordered non-cross but $x_1x_1\ta x_3x_2 \tb =x_1\ta x_1x_2 \tb x_3$ is still non-cross but not ordered. 

We continue with an example of very simple word equations whose minimal solution has exponential length, whose structure follows the one in \cite[Theorem 4.8]{KoscielskiP96}. 

\begin{proposition}\label{prop:exp-length} The minimal solution to the word equation $x_n \ta x_n \tb x_{n-1} \tb x_{n-2} \cdots  \tb x_1  = \ta x_n x_{n-1}^2 \tb x_{n-2}^2 \tb ...\tb x_1^2 \tb \ta^2$ has length $\Theta(2^n)$.
\end{proposition}

Finally, we recall the \textsc{$3$-Partition} problem (see \cite{gar:com}).
This problem is $\npclass$-complete in the strong sense, i.e., it remains $\npclass$-hard even when the input numbers are given in unary. 
\begin{problem}[\textsc{3-Partition} -- $\3par$]$ $\\
\emph{Instance:} $3m$ nonnegative integers (given in unary) $A=(k_1,\ldots,k_{3m})$, whose sum is $ms$\\
\emph{Question:} Is there a partition of $A$ into $m$ disjoint groups of three elements, such that each group sums exactly to $s$.
\end{problem}

\section{Lower bounds}

In this section, we show that the highly restricted class of regular-ordered word equations is $\npclass$-hard, and, thus, that even when the order in which the variables occur in an equation is fixed, and each variable may only repeat once -- and never on the same side of the equation -- satisfiability remains intractable. As mentioned in the introduction, our result shows the intractability of the satisfiability problem for a class of equations considerably simpler than the simplest intractable classes of equations known so far. Our result seems also particularly interesting since we are able to provide a corresponding upper bound in the next section, and even show that the minimal solutions of regular-ordered equations are ``optimally short''.

\begin{theorem}\label{sat_is_hard}
The satisfiability problem for regular-ordered word equations is $\npclass$-hard.
\end{theorem}

In order to show $\npclass$-hardness, we shall provide a reduction from a reachability problem for a simple type of regulated string-rewriting system. Essentially, given two words -- a starting point, and a target -- and an ordered series of $n$ rewriting rules (a rewriting program, in a sense), the problem asks whether this series of rules may be applied consecutively (in the predefined order) to the starting word such that the result matches the target. We stress that the order of the rules is predefined, but the place where a rule is to be applied within the sentential form is non-deterministically chosen. 
%

%

\begin{problem}[\textsc{Rewriting with Programmed Rules} -- $\rep$]$ $\\
\emph{Instance:} Words $u_{start},u_{end}\in \Sigma^*$ and an ordered series of $n$ substitution rules $w_i \to {w^\prime_i}$, with $w_i,w'_i\in \Sigma^*,$ for $1\leq i\leq n$.\\
\emph{Question:} Can $u_{end}$ be obtained from $u_{start}$ by applying each rule (i.e., replacing an occurrence of $w_i$ with ${w^\prime_i}$), in order, to $u_{start}$.
\end{problem}

\begin{example}
Let $u_{start}=b^5$ and $u_{end}=(a^{11}bc^2)^5$; for $1\leq i\leq 10$, consider the rules $w_i\to w'_i$ with $w_i=b $ and $w'_i=a^ibc$. We can obtain $u_{end}$ from $u_{start}$ by first applying $w_1\to w'_1$ to the first $b$, then $w_2\to w'_2$ to the second $b$, and further, in order for $3\leq i\leq 5$, by applying $w_i\to w'_i$ to the $i^{th}$ $b$. Then, we apply $w_6$ to the fifth $b$ (counting from left to right). Further we apply in order, for $7\leq i\leq 10$, $w_i\to w'_i$ to the $(11-i)^{th}$ occurrence of $b$.
\end{example}

It is not so hard to see that $\rep$ is $\npclass$-complete (the size of the input is the sum of the lengths of $u_{start},u_{end}, w_i$ and ${w^\prime_i}$). A reduction can be given from $\3par$, in a manner similar to the construction in the example above; important to our proof, $\3par$ is strongly $\npclass$-complete, so it is simpler to reduce it to a problem whose input consists~of~words.

\begin{lemma}\label{rep_is_hard}
$\rep$ is $\npclass$-complete.
\end{lemma}

Our reduction centres on the construction, for any instance $\mu$ of $\rep$, of a regular-ordered word equation $\alpha_\mu = \beta_\mu$ which possesses a specific form of solution -- which we shall call \emph{overlapping} -- if and only if the instance of $\rep$ has a solution. By restricting the form of solutions in this way, the exposition of the rest of the reduction is simplified considerably.
\begin{definition}
Let $n \in \mathbb{N}$, $\mu$ be an instance of $\rep$ with $u_{start}$, $u_{end}$ and rules $w_i \to {w^\prime_i}$ for $1\leq i \leq n$. Let $\#$ be a `new' letter not occurring in any component of $\rep$. We define the regular-ordered equation $\alpha_\mu = \beta_\mu$ such that:
\begin{align*}
\alpha_\mu &:= x_1\ w_1 \   x_2\ w_2  \  \cdots \  x_n \  w_n \  x_{n+1} \ \# \  u_{end},\\ 
\beta_\mu &:= \#\  u_{start} \  x_1 \  {w^\prime_1} \  x_2\  {w^\prime_2} \  x_3\ \  \cdots\ \  x_n\  {w}^\prime_n \  x_{n+1}.
\end{align*}
A solution $h : (X \cup \Sigma)^* \to \Sigma^*$ is called \emph{overlapping} if, for every $1 \leq i \leq n$, there exists $z_i$ such that $w_iz_i$ is a suffix of $h(x_i)$ and
$h(\# u_{start} x_1 \cdots {w^\prime_{i-1}} x_i) = h(x_1 w_1 \cdots x_i w_i) z_i.$
\end{definition}



Of course, satisfiability of a class of word equations asks whether any solution exists, rather than just overlapping solutions. Hence, before we prove our claim that $\alpha_\mu = \beta_\mu$ has an overlapping solution if and only if $\mu$ satisfies $\rep$, we present a construction of an equation $\alpha = \beta$ which has a solution if and only if $\alpha_\mu = \beta_\mu$ has an overlapping solution. Essentially, this shows that solving the satisfiability of regular-ordered equations is as hard as solving the satisfiability of word equations when we restrict our search to overlapping solutions.
\begin{lemma}\label{lem:sol3}
Let $\mu$ be an instance of $\rep$. There exists a regular-ordered equation $\alpha = \beta$ of size $O(|\alpha_\mu\beta_\mu|)$ such that $\alpha = \beta$ is satisfiable if and only if there exists an overlapping solution to $\alpha_\mu = \beta_\mu$.
\end{lemma}

The proof of the fact that the equation $\alpha_\mu = \beta_\mu$ has an overlapping solution if and only if $\mu$ satisfies $\rep$ has two main parts. The first is a slightly technical characterisation of overlapping solutions to $\alpha_\mu = \beta_\mu$ in terms of the periods $v_i$ of the images $h(x_i)$, which play a key role in modelling the `computation steps' of the rewriting system $\mu$.

\begin{lemma}\label{lem:solutions}
Let $\mu$ be a an instance of $\rep$ with $u_{start}, u_{end}$ and rules $w_i \to {w^\prime_i}$ for $1 \leq i \leq n$.
A substitution $h : (X \cup \Sigma)^* \to \Sigma^*$ is an overlapping solution to $\alpha_\mu = \beta_\mu$ if and only if there exist prefixes $v_1,v_2, \ldots, v_{m+1}$ of $h(x_1),h(x_2),\ldots, h(x_{n+1})$ such that:
\begin{enumerate}
\item{} $h(x_i) \; w_i $ is a prefix of $v_{i}^\omega$ for $1 \leq i \leq n$, and
\item{} $v_1 = \# u_{start}$, and for $2 \leq i \leq n$, $v_i = y_{i-1}{w^\prime_{i-1}}$, and
\item{} $y_n \; {w^\prime_n} \; h(x_{n+1}) = h(x_{n+1}) \; \# u_{end}$,
\end{enumerate}
where for $1 \leq i \leq n$, $y_i$ is the suffix of $h(x_i)$ of length $|v_i| - |w_i|$. 
\end{lemma}

We shall now take advantage of Lemma~\ref{lem:solutions} in order to demonstrate the correctness of our construction of $\alpha_\mu = \beta_\mu$ -- i.e., that it has an overlapping solution if and only if $\mu$ satisfies $\rep$. The general idea of the construction/proof is that for a solution $h$, the periods $v_i$ of the variables $h(x_i)$ -- which are obtained as the `overlap' between the two occurrences of $h(x_i)$ -- store the $i^{th}$ stage of a rewriting $u_{start} \to \ldots \to u_{end}$. In actual fact, this is obtained as the conjugate of $i$ starting with $\#$. Thus the solution-word, when it exists, stores a sort-of rolling computation history.

\begin{figure}[ht]
\vspace*{-0.3cm}
	\centering
  \includegraphics[width=\textwidth ]{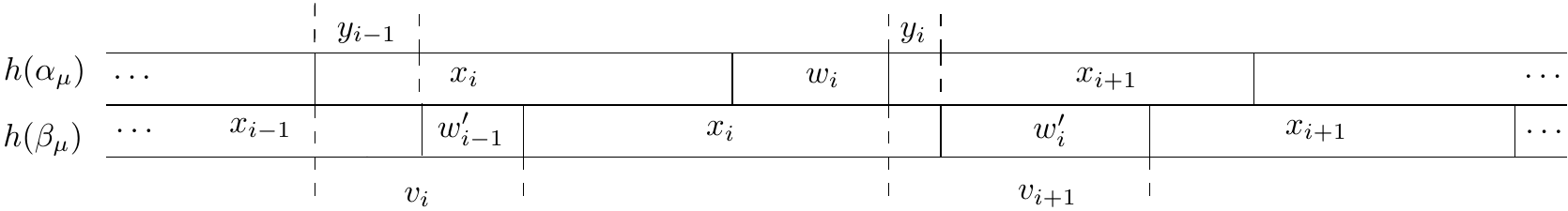}
	\caption{The periods of $h(x_i)$ in an overlapping solution to $\alpha_\mu = \beta_\mu$. The period of $h(x_i)$ is $v_i$, and since $w_i y_i$ is a suffix of $h(x_i)$ with the same length as $v_i$, we have that $w_i y_i$ is a cyclic shift of $v_i$ (i.e., they are conjugate) -- so $v_i = s w_i t$ and $(w_i) y_i = (w_i) t s$ for some $s,t$. $v_{i+1}$ is conjugate to $s {w^\prime_i} t$ since $v_{i+1} = y_i {w^\prime_i} = t s {w^\prime_i}$. Thus $v_{i+1}$ is obtained from $v_i$ by ``applying'' the rule $w_i \to {w^\prime_i}$.}	\label{Periods}
\vspace*{-0.3cm}
\end{figure}

\begin{lemma}\label{lem:sol2}
Let $\mu$ be a an instance of $\rep$ with $u_{start}, u_{end}$ and rules $w_i \to {w^\prime_i}$ for $1 \leq i \leq n$.
There exists an overlapping solution $h : (X \cup \Sigma)^* \to \Sigma^*$ to the equation $\alpha_\mu = \beta_\mu$ if and only if $\mu$ satisfies $\rep$.
\end{lemma}

\begin{proof}
Suppose firstly that $\mu$ satisfies $\rep$. Then there exist $s_1,s_2, \ldots, s_n, t_1, t_2, \ldots, t_n$ such that $ u_{start} = s_1 w_1 t_1$, for $1 \leq i \leq n$, $ s_i {w^\prime_i} t_i = s_{i+1} w_{i+1} s_{i+1}$ and $s_n {w^\prime_n} t_n = u_{end}$. Let $h: (X \cup \Sigma)^* \to \Sigma^*$ be the substitution such that $h(x_1) = \# s_1 w_1 t_1 \# s_1$, $h(x_{n+1}) = t_n \# s_n {w^\prime_n} t_n \# s_n {w^\prime_n} t_n$, and for $2 \leq i \leq n$,
 $h(x_i) = t_{i-1}\# s_{i-1} {w^\prime_{i-1}} t_{i-1} \# s_i$. 
 We shall now show that $h$ satisfies Lemma~\ref{lem:solutions}, and hence that $h$ is an overlapping solution to $\alpha_\mu = \beta_\mu$. 

Let $v_1 =  \# s_1 w_1 t_1$, let $y_1 := t_1 \# s_1$, and for $2 \leq i \leq n$, let $v_i := t_{i-1}\# s_{i-1} w_{i-1}$ and let $y_i :=  t_{i} \# s_i$. Let $y_n := t_n \# s_n$. Note that for $1 \leq i\leq n$, $v_i$ is a prefix of $h(x_i)$, and moreover, since $s_{i-1} {w^\prime_{i-1}} t_{i-1} =  s_i w_i t_i$, $y_i$ is the suffix of $h(x_i)$ of length $|v_i| - |w_i|$.

It is clear that $h$ satisfies Condition~(1) of Lemma~\ref{lem:solutions} for $i = 1$. For $2\leq i \leq n$, we have 
$h(x_i) w_i t_i = t_{i-1}\# s_{i-1} w_{i-1} t_{i-1} \# s_i w_i t_i = t_{i-1}\# s_{i-1} w_{i-1} t_{i-1} \# s_{i-1} w_{i-1} t_{i-1}$, which is a prefix of $v_i^\omega$, and hence $h(x_i) w_i$ is also a prefix of $v_i^\omega$. Since $v_i$ is also clearly a prefix fo $h(x_i)$, $h$ satisfies Condition~(1) for all $i$. Moreover, $v_1 = \# u_{start}$, and for $2 \leq i \leq n$, $y_{i-1} {w^\prime_{i-1}} = v_i = t_{i-1}\# s_{i-1} w_{i-1} =  v_i$, so $h$ satisfies Condition~(2). Finally,
\[y_n {w^\prime_n} h(x_{n+1}) = t_n \# s_n {w^\prime_n} t_n \# s_n {w^\prime_n} t_n \# s_n {w^\prime_n} t_n = h(x_{n+1}) \# u_{end}\]
so $h$ also satisfies Condition~(3).

Now suppose that $h$ is an overlapping solution to $\alpha_\mu = \beta_\mu$. Then $h$ satisfies Conditions~(1),~(2) and~(3) of Lemma~\ref{lem:solutions}. Let $v_i$, $y_i$ be defined according to the lemma for $1 \leq i \leq n$, and let $v_{n+1} = y_n {w^\prime_n}$. We shall show that $\mu$ satisfies $\rep$ as follows. We begin with the following observation.
\begin{claim}\label{sol2:claim1}
For $1 \leq i \leq n$, $y_i w_i $ and $v_i$ are conjugate. Hence, for $1 \leq i \leq n+1$, $|v_i|_\# = 1$.
\end{claim}
\begin{proof}[Proof (Claim \ref{sol2:claim1}).]
By Condition~(1) of Lemma~\ref{lem:solutions}, $h(x_i) w_i$ is a prefix of $v_i^\omega$. Since $y_i$ is the suffix of $h(x_i)$ of length $|v_i|-|w_i|$, this implies that $y_i w_i$ is a factor of $v_i^\omega$ of length $|v_i|$ and is therefore conjugate to $v_i$. By Condition~(2) of Lemma~\ref{lem:solutions} (and by definition, above, in the case of $i = n$), for $1 \leq i \leq n$, $v_{i+1} = y_{i} {w^\prime_{i}}$. Since $y_{i}w_{i}$ is conjugate to $v_{i}$ and $\#$ does not occur in either $w_i$ or ${w^\prime_i}$, it follows that $|v_{i+1}|_\# = |v_{i}|_\#$. Since $|v_1|_\# = |\# u_{start}|_\# =1$, the statement follows.
\end{proof}
Let $\tilde{v}_i$ be the (unique) conjugate of $v_i$ which has $\#$ as a prefix. We have the following important observation.
\begin{claim}\label{sol2:claim2}
For $1 \leq i \leq n$, there exist $s_i,t_i$ such that $\tilde{v}_i = \# s_i w_i t_i$ and $\tilde{v}_{i+1} = \# s_i {w^\prime_i} t_i$.
\end{claim}
\begin{proof}[Proof (Claim \ref{sol2:claim2}).]
By Claim~\ref{sol2:claim1}, $v_{i+1}$ contains an occurrence of $\#$ and by Condition~(2) of Lemma~\ref{lem:solutions}, $v_{i+1} = y_i {w^\prime_i}$ where $y_i$ is the suffix of $h(x_i)$ of length $|v_i| - |w_i|$. Note that since ${w^\prime_i}$ does not contain $\#$, it must occur at least once in $y_i$. Let $t_i$ be the (proper) prefix of $y_i$ up to the first occurrence of $\#$, and let $s_i$ be the corresponding suffix, so that $y_i = t_i \# s_i$. Then by Condition~(2) of Lemma~\ref{lem:solutions}, $v_{i+1} = t_i \# s_i {w^\prime_i}$, so $\tilde{v}_{i+1} = \# s_i {w^\prime_i} t_i$. Moreover, by Claim~\ref{sol2:claim1}, $y_i w_i = t_i \# s^\prime_i w_i$ is conjugate to $v_i$ and it follows that $\tilde{v}_i = \# s^\prime_i w_i t_i$.
\end{proof}

Recall from Condition~(3) of Lemma~\ref{lem:solutions} that $y_n {w^\prime_n} h(x_{n+1}) = v_{n+1} h(x_{n+1}) = h(x_{n+1}) \# u_{end}$. Consequently, $v_{n+1}$ and $\# u_{end}$ are conjugate, so $\tilde{v}_{n+1} = \# u_{end}$. Moreover, by Condition~(2) of Lemma~\ref{lem:solutions}, $v_1 = \tilde{v}_1 = \#u_{start}$. Thus, it follows from Claim~\ref{sol2:claim2} that $\mu$ satisfies $\rep$.\end{proof}

As it is clear that the equation $\alpha_\mu = \beta_\mu$ (and hence also the equation $\alpha =\beta$ given in Lemma~\ref{lem:sol3}) may be constructed in polynomial time, our reduction from $\rep$ is complete. So, by Lemmas \ref{rep_is_hard} and \ref{lem:sol2}, we have shown Theorem \ref{sat_is_hard}. 

\section{$\npclass$-upper bound}
In this section, we show that the satisfiability of  regular-ordered word equations is in $\npclass$.
\begin{theorem}\label{cor:inNP} 
The satisfiability problem for regular-ordered equations is in $\npclass$.
\end{theorem}

In order to achieve this, we extend the classical approach of filling the positions (see e.g.,~\cite{kar:the} and the references therein). This method essentially comprises of assuming that for a given equation $\alpha = \beta$, we have a solution $h$ with specified lengths $|h(x)|$ for each variable $x$. The assumption that $h$ satisfies the equation induces an equivalence relation on the positions of each $h(x)$: if a certain position in the solution-word is produced by an occurrence of the $i^{th}$ letter of $h(x)$ on the RHS and an occurrence of the $j^{th}$ letter of $h(y)$ on the LHS, then these two positions must obviously have the same value/letter  and we shall say that these occurrences \emph{correspond}. These individual equivalences can be combined to form equivalence classes, and if no contradictions occur (i.e., two different terminal symbols $\mathtt{a}$ and $\mathtt{b}$ do not belong to the same class), a valid solution can be derived.

\begin{figure}[ht]
\vspace*{-0.3cm}
	\centering
  \includegraphics[width=\textwidth ]{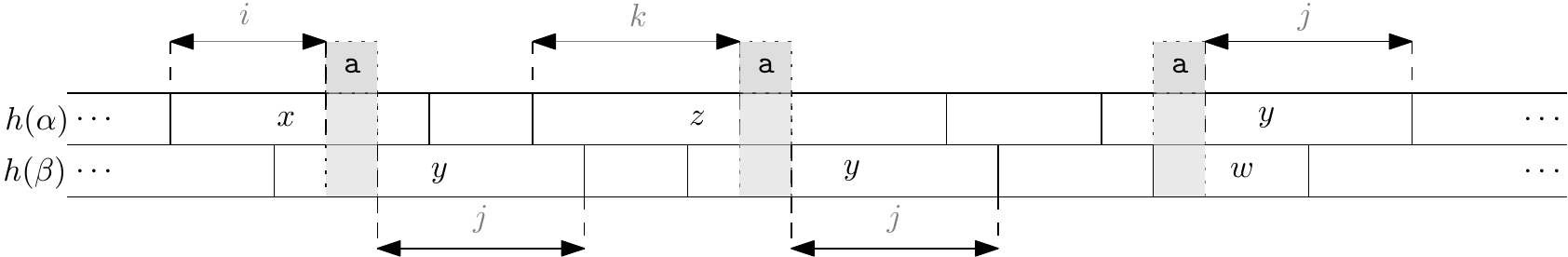}
	\caption{Fixing positions: since an occurrence of the $i^{th}$ letter of $h(x)$ corresponds to an occurrence of the $(|h(y)|-j)^{th}$ letter of $y$, whose other occurrences correspond to the $k^{th}$ letter of $h(z)$ and first letter of $h(w)$, all these positions are equivalent and contain the same letter, e.g.,~$\ta$.
	}
	\label{fig2}
\vspace*{-0.3cm}
\end{figure}

Such an approach already allows for some straightforward observations regarding the (non-)minimality of a solution $h$. In particular, if an equivalence class of positions is not associated with any terminal symbol, then all positions in this class can be mapped to $\varepsilon$, resulting in a strictly shorter solution. On the other hand, even for our restricted setting, this observation is insufficient to provide a bound on the length of minimal solutions. In fact, in the construction of the equivalence classes we ignore, or at least hide, some of the structural information about the solution. In what follows, we shall see that by considering the exact `order' in which positions are equated, we are able to give some more general conditions under which a solution is not minimal. 

Our approach is, rather than just constructing these equivalence classes, to construct sequences of equivalent positions, and to then analyse similar sequences. For example, one occurrence of a position $i$ in $h(x)$ might correspond to an occurrence of position $j$ in $h(y)$, while another occurrence of position $j$ in $h(y)$ might correspond to position $k$ in $h(z)$, and so on, in which case we would consider the sequence: $ \ldots \to (x,i) \to (y,j) \to (z,k) \to \ldots.$

The sequence terminates when either a variable which occurs only once or a terminal symbol is reached. For general equations, considering all such sequences leads naturally to a graph structure where the nodes are positions $(x,i) \in X\times \mathbb{N}$, and number of edges from each node is determined by the number of occurrences of the associated variable. Each connected component of such a graph corresponds to an equivalence class of positions as before. In the case of quadratic (and therefore also regular) equations, where each variable occurs at most twice, each `node' $(x,i)$ has at most two edges, and hence our graph is simply a set of disjoint chains, without any loops. As before, each chain (called in the following sequence) must be associated with some occurrence of a terminal symbol, which must occur either at the start or the end of the chain. Hence we have $k<n$ sequences, where $n$ is the length of the equation, such that every position $(x,i)$ where $x$ is a variable in our equation and $1 \leq i \leq |h(x)|$ occurs in exactly one sequence. It is also not hard to see that the total length of the sequences is upper bounded by $2|h(\alpha)|$.

In order to be fully precise, we will distinguish between different occurrences of a variable/terminal symbol by associating each with an index $z \in \mathbb{N}$ by enumerating occurrences from left to right in $\alpha\beta$. Of course, when considering quadratic equations, $z \in \{1,2\}$ for each variable $x$. Formally, we define our sequences for a given solution $h$ to a quadratic equation $\alpha = \beta$ as follows: a \emph{position} is a tuple $(x,z,d)$ such that $x$ is a variable or terminal symbol occurring in $\alpha\beta$, $1 \leq z \leq |\alpha\beta|_x$, and $1 \leq d \leq |h(x)|$. Two positions $(x,z,d)$ and $(y,z^\prime,d^\prime)$ \emph{correspond} if they generate the same position in the solution-word. The positions are \emph{similar} if they belong to the same occurrence of the same variable (i.e., $x=y$ and $z = z^\prime$). For each position $p$ associated with either a terminal symbol or a variable occurring only once in $\alpha\beta$, we construct a sequence $S_p = p_1,p_2,\ldots$ such that
\begin{itemize}
\item{}$p_1 = p$ and $p_2$ is the (unique) position corresponding with $p_1$, and 
\item{}for $i \geq 2$, if $p_i = (x,z,d)$ such that $x$ is a terminal symbol or occurs only once in $\alpha\beta$, then the sequence terminates, and
\item{}for $i \geq 2$, if $p_i = (x,z,d)$,  such that $x$ is a variable occurring twice, then $p_{i+1}$ is the position corresponding to the (unique) position $(x,z^\prime,d)$ with $z^\prime \not= z$ (i.e., the `other' occurrence of the $i^{th}$ letter in $h(x)$). 
\end{itemize}
We extend the idea of similarity from positions to sequences of positions in the natural way: two sequences $p_1,p_2,\ldots,p_i$ and $q_1,q_2,\ldots q_i$ are \emph{similar} whenever $p_j$ and $q_j$ are similar for all $j \in \{1,2,\ldots, i\}$. Our main tool is the following lemma, which essentially shows that if a sequence contains two similar consecutive subsequences (so, a square), then the solution defining that sequence is not minimal. 
\begin{figure}[ht]
\vspace*{-0.3cm}
	\centering
  \includegraphics[width=\textwidth ]{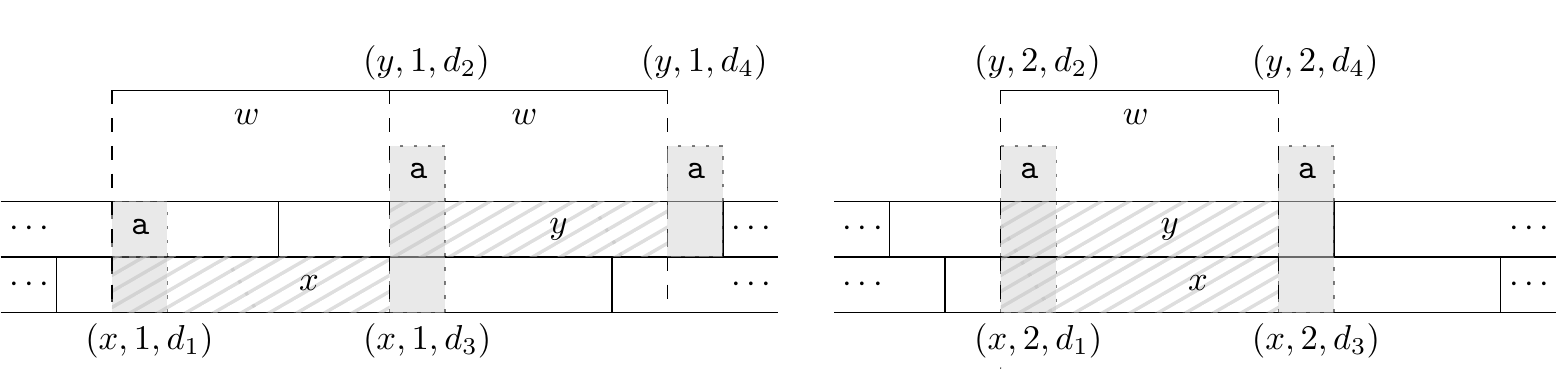}
	\caption{Illustration of Lemma \ref{lem:squares1} in the case of a short subsequence $\ldots, (x,1,d_1), (y,2,d_2), (x,1,d_{3}), (y,2,d_{4}), \ldots$: since the two sequences starting at $(x,1,d_1)$ and $(x,1,d_{3})$ are similar, they define a common region $w$ (shaded). Since they are consecutive, the first and last occurrences of $w$ are adjacent, and on opposite sides of the equation. Thus, removing the region $w$ from $h(x)$ and $h(y)$ does not alter the fact that $h$ satisfies the equation.}
	
	\label{fig3}
\vspace*{-0.3cm}
\end{figure}

\begin{lemma}\label{lem:squares1}
Let $h$ be a solution to a quadratic equation $\alpha = \beta$, and let $p$ be a position associated with a single-occurring variable or terminal symbol. If the sequence $S_p$ has a subsequence $p_1,p_2,\ldots p_t,p_{t+1},p_{t+2},\ldots,p_{2t}$ such that $p_1,p_2,\ldots p_t$ and $p_{t+1},p_{t+2},\ldots,p_{2t}$ are similar, then $h$ is not minimal.
\end{lemma} 

\begin{proof}
Assume that $S_p$ has such a subsequence and assume w.l.o.g.\ that it is length-minimal (so $t$ is chosen to be as small as possible). For $1 \leq i \leq 2t$, let $p_i = (x_i, z_i, d_i)$ and note that by definition of similarity, for $1 \leq i \leq t$, $x_i = x_{i+t}$ and $z_i = z_{i+t}$. Assume that $d_1 < d_{t+1}$ (the case that $d_1 > d_{t+1}$ may be treated identically).
\begin{claim}\label{squares1:claim1}\emph{
Suppose that $(x,z,i), (x,z,i^\prime), (x,z,i^{\prime\prime})$ are positions with $i<i^\prime<i^{\prime\prime}$ such that $(x,z,i)$ and $(x,z,i^{\prime\prime})$ correspond to $(y,z^\prime,j)$ and $(y,z^\prime,j^{\prime\prime})$ respectively. Then $j^{\prime\prime} - j = i^{\prime\prime} - i$, and there exists $j^\prime$ with $j^\prime - j = i^\prime - i$ such that $(x,z,i^\prime)$ and $(y,z^\prime,j^\prime)$ correspond.}
\end{claim}

\begin{proof}[Proof (Claim \ref{squares1:claim1})]
Follows directly from the fact that there is only one occurrence of $x$ with associated index $z$ and only one occurrence of $y$ with associated index $z^\prime$.
\end{proof}

A straightforward consequence of Claim~1 is that there exists a constant $C\in\mathbb{N}$ such that for all $i\in \{1,2,\ldots, t\}$, $d_{i+t} - d_i = C$. Intuitively, each pair of similar positions $p_i = (x_i,z_i,d_i)$ and $p_{i+t} = (x_i,z_i,d_i+C)$ are the end positions of a factor $h(x)[d_i.. d_i+C-1]$, which as we shall see later on in the proof, can be removed to produce a shorter solution $g$. 

We can also infer from Claim~1 that for positions $(x,z,i), (x,z,i^\prime)$, $(x,z,i^{\prime\prime})$ with $i<i^\prime<i^{\prime\prime}$, if the subsequences of length $n$ beginning with $(x,z,i)$ and $(x,z,i^{\prime\prime})$ are similar, then so are the subsequences of length $n$ beginning with $(x,z,i)$ and $(x,z,i^\prime)$. It follows that the subsequence does not contain a position `between' $(x_1,z_1,d_1)$ and $(x_1,z_1,d_1+C)$ (and likewise for $(x_t,z_t,d_t)$ and $(x_t,z_t,d_t+C)$, and hence that the respective factors $h(x_1)[d_1.. d_1+C-1]$ and $h(x_t)[d_t.. d_t+C-1]$ do not overlap with other such factors, which will be useful later.

\begin{claim}\label{squares1:claim2}\emph{
Let $j \in \{2,\ldots, t, t+2, \ldots, 2t\}$ such that $x_j = x_1 (=x_{t+1})$ and $z_j = z_1 (=z_{t+1})$. Then $d_j \notin \{d_1,\ldots, d_{t+1} (= d_1+C)\}$. Likewise, if $j \in \{1,\ldots, t-1, t+1, \ldots 2t-1\}$ such that $x_j = x_t$ and $z_j = z_t$, then $d_j \notin \{d_t, \ldots, d_{2t}\}$.
}
\end{claim}

\begin{proof}[Proof (Claim \ref{squares1:claim2})]
We prove the statement for $x_j = x_1$. The case that $x_j = x_t$ holds symmetrically. Suppose to the contrary that $x_j = x_1, z_j = z_1$ and $d_j \in \{d_1,\ldots, d_{t+1}\}$. Clearly $d_j \notin \{d_1,d_{t+1}\}$, otherwise the sequence contains the same position twice and is therefore an infinite cycle which contradicts the definition. Then by Claim~1, since the sequences of length $t$ beginning with $(x_1,z_1,d_1)$ and $(x_1,z_1,d_{t+1})$ are similar, the sequences of length $t$ beginning with $(x_1,z_1,d_1)$, $(x_1,z_1,d_j)$ and $(x_1,z_1,d_{t+1})$ are pairwise similar. However, $(x_1,z_j,d_j) (= (x_j,z_j,d_j))$ is contained in either the sequence of length $t$ beginning with $(x_1,z_1,d_1)$ or with $(x_1,z_1,d_{1+t})$. In both cases, we get a shorter subsequence $p^\prime_1,p^\prime_2,\ldots p^\prime_{t^\prime},p^\prime_{t^\prime+1},p^\prime_{t^\prime+2},\ldots,p^\prime_{2t^\prime}$ such that $p^\prime_1,p^\prime_2,\ldots p^\prime_{t^\prime}$ and $p^\prime_{t^\prime+1},p^\prime_{t^\prime+2},\ldots,p^\prime_{2t^\prime}$ are similar. This contradicts our assumption that $t$ is as small as possible.
\end{proof}

We are now ready for the main argument of the proof. Using the observations above, we shall remove parts of the solution $h$ to obtain a new, strictly shorter solution and thus show that $h$ is not minimal as required. To do this, we shall define a new equation $\alpha^\prime = \beta^\prime$ obtained by replacing the second occurrence of each variable $x$ (when it exists) with a new variable $x^\prime$. We note a few obvious facts. Firstly, we can derive a solution $h^\prime$ to $\alpha^\prime = \beta^\prime$ from the solution $h$ to our original equation by simply setting $h^\prime(x) = h^\prime(x^\prime) =  h(x)$ for all $x \in \var(\alpha\beta)$. Likewise, any solution to $\alpha^\prime =\beta^\prime$ for which this condition holds (i.e., $h^\prime(x) = h^\prime(x^\prime)$ for all $x \in \var(\alpha\beta)$) induces a solution $g$ to our original equation $\alpha = \beta$ given by $g(x) = h^\prime(x) (= h^\prime(x^\prime))$. Finally, for each position $(x,z,d)$ in the original solution $h$, there exists a unique ``associated position'' in $h^\prime$ given by $h(x)[d]$ if $z= 1$ and $h(x^\prime)[d]$ if $z = 2$. Furthermore, it follows from the definitions that for any pair of positions $p,q$ which correspond (in terms of $h$), we can remove the associated positions from $h^\prime$ and the result will still be a valid solution to our modified equation $\alpha^\prime = \beta^\prime$ (although such a solution may no longer induce a valid solution to our original equation, since it is no longer necessarily the case that $h(x) = h(x^\prime)$ for all $x$). 

We construct our shorter solution $g$ to $\alpha = \beta$ as follows. Let $h^\prime$ be the solution to $\alpha^\prime = \beta^\prime$ derived from $h$.
Recall from the definition of $S_p$ that, for $1\leq i < t$, the positions $(x_i, \overline{z}_i,d_i)$ and $(x_{i+1},z_{i+1},d_{i+1})$ correspond, where $\overline{z} = (z + 1) \mod 2$ (i.e., so that $\overline{z} \not= z$). Moreover, $(x_i,\overline{z}_i,d_{i}+C)$ and $(x_{i+1},z_{i+1},d_{i+1}+C)$ correspond, and thus by Claim~1, $(x_i,\overline{z}_i,d_i + k)$ and $(x_{i+1},z_{i+1}, d_{i+1} + k)$ correspond for $0 \leq k \leq C-1$. Since corresponding positions must have the same value/letter, it follows that there exists a factor $w \in \Sigma^+$ such that $w = h(x_i)[d_i.. d_i + C-1] (= h^\prime(x)[d_i.. d_i + C-1] = h^\prime(x^\prime)[d_i.. d_i + C-1])$ for $1\leq i \leq t$.

For each corresponding pair of positions $(x_i,\overline{z}_i,d_i + k), (x_{i+1},z_{i+1}, d_{i+1} + k)$ such~that $0 \leq k \leq C-1$ and $1 \leq i \leq t-1$, delete the associated positions in $h^\prime$ to obtain a new solution $h^{\prime\prime}$ to $\alpha^\prime = \beta^\prime$. Thus, for every position associated with $(x_i,\overline{z}_i,d_i+k)$ such~that $1 < i \leq t$, we also delete the position associated with $(x_i, z_i, d_i +k)$. Hence, for all $x \notin \{x_1,x_t\}$, $h^{\prime\prime}(x) = h^{\prime\prime}(x^\prime)$. In order to guarantee that $h^{\prime\prime}(x) = h^{\prime\prime}(x^\prime)$ for $x \in \{x_1,x_t\}$, we must also delete the positions associated with $(x_1,z_1,d_1+k)$ and $(x_t,\overline{z}_t,d_t+k)$ for $0\leq k < C$. To~see that, in doing so, we still have a valid solution to $\alpha^{\prime} = \beta^\prime$, note firstly that, by Claim~2, we have not deleted any of these positions already. Moreover, it follows from the sequence $S_p$ that $(x_t,\overline{z}_t,d_t)$ corresponds to $(x_1,z_1,d_1+C)$. Assume $z_1 = 1$ (the~case~$z_1 = 2$ is symmetric). It follows that $z_t = 1$ (since $\overline{z}_t \not= z_1$). Thus there exists an index $m$ such that $h^{\prime\prime}(x_1)[d_1.. d_1+ C-1]$ generates the factor $w$ starting at position $m$ in $h^{\prime\prime}(\alpha^\prime)$ and $h^{\prime\prime}(x_t)[d_t.. d_t+ C-1]$ generates the (same) factor $w$ starting at position $m+|w|$ in $h^{\prime\prime}(\beta)$. It is straightforward to see that removing these factors (i.e., deleting the positions associated with $(x_1,z_1,d_1+k)$ and $(x_t,\overline{z}_t,d_t+k)$ for $0\leq k \leq C-1$) does not affect the agreement of the two sides of the equation. Thus we obtain a shorter solution $h^{\prime\prime}$ to $\alpha^\prime = \beta^\prime$ such that $h(x) = h(x^\prime)$ for all variables $x$, hence a shorter solution $g$ given by $g(x) = h^{\prime\prime}(x)$~to~$\alpha=\beta$.
\end{proof}

Using Lemma~\ref{lem:squares1}, we obtain as a direct consequence that minimal solutions to regular-ordered equations are at most linear in the length of the equation.

\begin{proposition}\label{prop:linearsolutions}
Let $E$ be a regular-ordered word equation with length $n$, and let $h$ be a minimal solution to $E$. Then $|h(x)| < n$ for each variable $x$ occurring in $E$.
\end{proposition}

\begin{proof}
Firstly, we note that for a minimal solution $h$ to $E$, every position of $h$ occurs somewhere in one of the associated sequences $S_p$. Since there can be no more than $n$ such sequences, it is sufficient to show that each one contains at most one position $(x,z,d)$ for each variable $x$. Let $h$ be a minimal solution to $E$ and let $S_p$ be any sequence. Firstly, we note that $S_p$ does not contain a subsequence $(x,z,d), (x,z^\prime,d^\prime)$. In particular, if such a subsequence existed, then since $E$ is regular, we would have $z = z^\prime$, and Lemma~\ref{lem:squares1} would imply a contradiction. Now consider a subsequence $(x,z,d),(x^\prime,z^\prime,d^\prime),(x^{\prime\prime},z^{\prime\prime},d^{\prime\prime})$. By definition, this implies that $(x,\overline{z},d)$ corresponds to $(x^\prime,z^\prime,d^\prime)$, and that $(x^\prime, \overline{z}^\prime, d^\prime)$ corresponds to $(x^{\prime\prime},z^{\prime\prime},d^{\prime\prime})$. Suppose that $x$ occurs to the left of $x^\prime$ in $E$ (and note that since $E$ is regular-ordered, this holds for both sides of the equation). Then $(x,\overline{z},d)$ occurs to the left of $(x^\prime,\overline{z}^\prime,d^\prime)$. Since they correspond, it follows that $(x^\prime,z^\prime,d^\prime)$ occurs to the left of $(x^{\prime\prime},z^{\prime\prime},d^{\prime\prime})$, and thus that $x^\prime$ occurs to the left of $x^{\prime\prime}$. Since $x \not= x^\prime$ and $x^\prime \not= x^{\prime\prime}$, it is clear by iteratively applying this argument that each further position in the sequence must belong to a new variable occurring further right in $E$, and our statement holds. The case that $x$ occurs to the right of $x^\prime$ may be treated symmetrically.
\end{proof}

We can see that, in terms of restricting the lengths of individual variables, the result in Proposition~\ref{prop:linearsolutions} is optimal. For instance, in a minimal solution $h$ to the equation
 $w \tc x_1 = x_1 \tc w $, with $w\in \{a,b\}^*$ ,the variable $x_1$ is mapped to~$w $, so $|h(x)|=|E|-2 \in O(|E|)$. 
Furthermore, Theorem \ref{cor:inNP} follows now as a direct consequence of Proposition~\ref{prop:linearsolutions} and Lemma~\ref{lem:shortsolutions}, as the length of a minimal solution to a regular-ordered equation $\alpha=\beta$ is $O(|\alpha\beta|^2)$.

Note that it is a simple consequence of Proposition~\ref{prop:linearsolutions} that the satisfiability of a regular-ordered equation $E$ with a constant number $k$ of variables can be checked in $\pclass$-time: we guess the length ($\leq |E|$) of the image of each variable in the minimal solution, and then it can be checked in $\pclass$-time whether a solution with these lengths actually~exists.
\section{Tractable equations}

Finally, we discuss a class of equations for which satisfiability is in $\pclass$. Tractability was obtained so far from two sources: bound the number of variables by a constant (e.g., one or two-variable equations~\cite{Jez2016,DabPla2004}), or heavily restrict their structure (e.g., regular equations whose sides do not have common variable, or equations that only have one repeated variable, but at least one non-repeated variable on each side \cite{ManSchNow16}). 

The class we consider slightly relaxes the previous restrictions. As the satisfiability of quadratic or even regular-ordered equations is already $\npclass$-hard it seems reasonable to consider here patterns where the number of repeated variables is bounded by a constant (but may have an arbitrary number of non-repeated variables). More precisely, we consider here non-cross equations with only one repeated variable. This class generalises naturally the class of one-repeated variables.

\begin{theorem}\label{prop:1repvar}
Let $x\in X$ be a variable and $\mathcal{D}$ be the class of word equations $\alpha = \beta$ such that $\alpha,\beta\in (\Sigma\cup X)^*$ are non-cross and each variable of $X$ other than $x$ occurs at most once in $\alpha\beta$. Then the satisfiability problem for $\mathcal{D}$ is in $\pclass$.
\end{theorem}

In the light of the results from \cite{ManSchNow16}, it follows that the interesting case of the above theorem is when the equation $\alpha = \beta$ is such that $\alpha = x u_1 x u_2 \cdots u_k x$ and $\beta = \beta^{\prime}v_0 x v_1 x v_2 \cdots x v_k \beta^{\prime\prime}$ where $v_0,v_1,\ldots v_k, u_1, u_2, \ldots u_k \in \Sigma^*$ and $\beta^{\prime},\beta^{\prime{\prime}}$ are regular patterns that do not contain $x$ and are variable disjoint. Essentially, this is a matching problem in which we try to align two non-cross patterns, one that only contains a repeated variable and constants, while the other contains the repeated variable, constants, and some wild-cards that can match any factor. The idea of our proof is to first show that such equations have minimal solutions~of~polynomial length. Further, we note that if we know the length of $\beta^{\prime}$ (w.r.t. the length of $\alpha$) then we can determine the position where the factor $v_0 x v_1 x v_2 \cdots x v_k$ occurs in $\alpha$, so the problem boils down to seeing how the positions of $x$ are fixed by the constant factors $v_i$. Once this is done, we check if there exists an assignment of the variables of $\beta^{\prime}$ and $\beta^{\prime\prime} $ such that the constant factors of these patterns fit correctly to the corresponding prefix, respectively, suffix of $\alpha$. 

%
%

\section{Conclusions and Prospects}

The main result of this paper is the $\npclass$-completeness of the satisfiability problem for regular-ordered equations. While the lower bound seems remarkable to us because it shows that solving very simple equations, which also always have  short solutions, is $\npclass$-hard, the upper bound seems more interesting from the point of view of the tools we developed to show~it. We expect the combinatorial analysis of sequences of equivalent positions in a minimal solution to an equation (which culminated here in Lemma \ref{lem:squares1}) can be applied to obtain upper bounds on the length of the minimal solutions to more general equations than just the regular-ordered ones. It would be interesting to see whether this type of reasoning leads to polynomial upper bounds on the length of minimal solutions to regular (not ordered) or quadratic equations, or to exponential upper bounds on the length of minimal solutions of non-cross or cubic equations. In the latter cases, a more general approach should be used, as the equivalent positions can no longer be represented as linear sequences, but rather as directed graphs.


Lemma \ref{lem:squares1} helps us settle the status of the satisfiability problem for regular-ordered equations with regular constraints. This problem is in $\npclass$, when the languages defining the scope of the variables are all accepted by finite automata with at most $c$ states, where $c$ is a~constant, as well as in the case egular-ordered equations whose sides contain exactly the same variables (see the proofs in Appendix). The satisfiability problem for regular-ordered equations with general regular constraints still remains $\pspaceclass$-complete.

Regarding the final section our paper, it seems interesting to us to see whether deciding the satisfiability of word equations with one repeated variable (so without the non-cross sides restriction) is still tractable. Also, it seems interesting to analyse the complexity of word equations where the number of repeated variables is bounded by a constant. 
\newpage
\bibliography{equations}

\newpage
\section*{Appendix}
\noindent {\bf Proof of Proposition \ref{prop:exp-length}}:
\begin{proof}
The minimal (and single) solution to the equation maps $x_i$ to $\ta^{2^i}$. Indeed, $x_n$ must be mapped to $\ta^\ell$ for some $\ell$, and none of the variables $x_i$, with $1\leq i\leq n-1$ can be mapped to a word containing $\tb$ (or the number of $\tb$'s would be greater in the image of the RHS). So, $x_1$ will be mapped to $a^2$, $x_2=x_1^2$ to $a^4$, and, in general, $x_{i+1}=x_i^2$, for $1\leq i\leq n-1$. The conclusion follows.
\end{proof}

\noindent {\bf Proof of Lemma \ref{rep_is_hard}}:
\begin{proof}
Let $S = (k_1,k_2,\ldots, k_{3m})$ be an instance of $\3par$ with $k_i \in \mathbb{N}$ for $1 \leq i \leq n$. Let $s := \frac{1}{m} \sum\limits_{i = 1}^n k_i$. We construct an instance $\mu$ of $\rep$ as follows. Let $u_{start} := \tb^m$ and let $u_{end} := {(\ta^s\tb\tc^3)}^m$. For $1 \leq i \leq n$, let $w_i := \tb$ and let ${w^\prime_i} := \ta^{k_i}\tb\tc$. Since $S$ is given in unary, $\mu$ can be constructed in polynomial time.

Suppose firstly that $S$ satisfies $\3par$. Associate with each subset in the partition a number from $1$ to $m$, and let $d_i$ be the number associated with the subset in which $k_i$ is placed in the partition. To see that $\mu$ satisfies $\rep$, apply the rewriting rules by swapping the ${d_i}^{th}$ occurrence of $\tb$ (i.e. $w_i$) with $w^\prime_i = \ta^{k_i}\tb\tc$. Note that applying each rule in this manner increases the number of $\ta$-s to the left of the ${d_i}^{th}$ occurrence of $\tb$ by $k_i$, and the number of $\tc$-s to the right by $1$. More formally, if the word before applying the rule $w_i \to {w^\prime_i}$ is: 
\[\ta^{p_1}\tb\tc^{q_1}\ta^{p_2}\tb\tc^{q_2}\cdots \ta^{p_{d_i}} \tb \tc^{q_{d_i}} \cdots \ta^{p_m}\tb\tc^{q_m}\]
then the word after applying the rule is:
\[ \ta^{p_1}\tb\tc^{q_1}\ta^{p_2}\tb\tc^{q_2} \cdots\ta^{p_{d_i}+k_i} \tb \tc^{q_{d_i}+1} \cdots \ta^{p_m}\tb\tc^{q_m}.\]
Thus, after applying all the rules, we get a word:
\[ u = \ta^{p_1}\tb\tc^{q_1}\ta^{p_2}\tb\tc^{q_2} \cdots \ta^{p_m}\tb\tc^{q_m}.\]
such that $p_i = \sum\limits_{d_j = i}k_j $ and $q_i = \sum\limits_{d_j = i} 1$. It follows from the fact that $S$ satisfies $\3par$ that, for $1\leq i \leq m,\sum\limits_{d_j = i}k_j = s$ and $\sum\limits_{d_j = i} 1 = 3$. Thus $ u= u_{end}$ and $\mu$ satisfies $\rep$.

Now suppose that $\mu$ satisfies $\rep$. Then there exist a series of indexes $d_1,d_2,\ldots, d_n$ such that consecutively replacing the ${d_i}^{th}$ occurrence of $\tb$ in $u_{start}$ produces the result $u_{end} = (\ta^s\tb\tc^3)^m$. By the same reasoning as above, this implies that $\sum\limits_{d_j = i}k_j = s$, and $\sum\limits_{d_j = i} 1 = 3$. Consequently, it can be observed by partitioning $S$ into subsets $S_1,\ldots S_m$ such that $k_j \in S_i$ if and only if $d_j = i$, that each subset $S_i$ contains 3 elements which sum to $s$, and thus that $S$ satisfies $\3par$.

To conclude this proof, it is immediate to note that $\rep$ is in $\npclass$. 
\end{proof}

\noindent {\bf Proof of Lemma \ref{lem:sol3}}:
\begin{proof}
Let $u_{start},u_{end}$ and $w_i \to {w^\prime_i}$ for $1\leq i \leq n$ be the relevant parts of $\mu$. Let 
\begin{align*}
\alpha &:= x_1 \# x_2 \# x_3 \; w_1 \; x_4 \# x_5 \# x_6  \; w_2  \; \cdots \; x_{3n+1} \# x_{3n + 2} \# x_{3n+3} \; \# u_{end},\\ 
\beta &:= \#\, u_{start} \; x_1 \# x_2 \# x_3 \, {w^\prime_1}  x_4 \# x_5 \# x_6\;\cdots\; {w^\prime_n} x_{3n+1} \# x_{3n+2} \# x_{3n+3}.
\end{align*}
Now, suppose there exists an overlapping solution $h : (X \cup \Sigma)^* \to \Sigma^*$ to $\alpha_\mu = \beta_\mu$, and for $1 \leq i \leq n+1$, let $v_i$ be the prefix of $h(x_i)$ in accordance with Lemma~\ref{lem:solutions}. It is clear that the conditions of Lemma~\ref{lem:solutions} are also satisfied by the substitution $h^\prime$ given by $h^\prime(x_i) = v_{i} h(x_i)$, and thus that $h^\prime$ is also an overlapping solution to $\alpha_\mu = \beta_\mu$. It follows from Claim~(1) of Lemma~\ref{lem:sol2} that $|v_i|_\# = 1$, and from the definition of $v_i$ that $h(x_i)$ has $v_i$ as a prefix. Hence $v_i h(x_i)$ contains at least two occurrences of $\#$, so there exist $z,z^\prime, z^{\prime\prime} \in \Sigma^*$ such that $h^\prime(x_i) = z \# z^\prime \# z^{\prime\prime}$. It is straightforward that the substitution $g : (X \cup \Sigma)^* \to \Sigma^*$ given by $g(x_{3i-2}) := z$, $g(x_{3i-1}) := z^\prime z$ and $g(x_{3i}) := z^{\prime\prime}$ is a solution to $\alpha = \beta$.

Now suppose instead that there exists a solution $g : (X \cup \Sigma)^* \to \Sigma^*$ to $\alpha = \beta$. Let $h : (X\cup \Sigma)^* \to \Sigma^*$ be the substitution given by $h(x_i) := g(x_{3i-2}) \# g(x_{3i-1}) \# g(x_{3i})$. Clearly, $h$ is a solution to $\alpha_\mu = \beta_\mu$. Thus it remains to show that it is overlapping, which we can do by counting the occurrences of $\#$. In particular, note that for $1 \leq i \leq n$, 
\[|h(x_1 w_1 \cdots w_{i-1} x_i w_i)|_\# = |h(\# u_{start} x_1 {w^\prime_1} \cdots x_i)|_\# -1.\]
Since $|h(x_i)|_\# \geq 2$, and $|w_i|_\# = 0$, the penultimate $\#$ in $x_i$ on the RHS must correspond to the last $\#$ on the LHS. More formally, there exist $s_1,s_2,s_3$ such that $h(x_i) = s_1\#s_2\#s_3$  (with $|s_2|_\# = |s_3|_\# = 0$) such that:
\[h(x_1 w_1 \cdots w_{i-1}) s_1 \# s_2  \#= h(\# u_{start} x_1 {w^\prime_1} \cdots {w_{i-1}}^\prime) s_1 \#.\]
Hence the suffix $s_2 \# s_3$ of $h(x_i)$ has $w_i$ as a factor, and for $1 \leq i \leq n$, there exists $z_i$ such that $w_iz_i$ is a suffix of $h(x_i)$ and $h$ is an overlapping solution to $\alpha_\mu = \beta_\mu$.
\end{proof}

\noindent {\bf Proof of Lemma \ref{lem:solutions}}:
\begin{proof}
Let $h : (X\cup \Sigma)^* \to \Sigma^*$ be a substitution. Suppose firstly that $h$ satisfies the conditions of the lemma. It can easily be determined (cf. Claim~1 in the proof of Lemma~\ref{lem:sol2}) that for $1 \leq i \leq n$ that $|v_i|_\# = 1$. By Conditions~(1) and~(2), $h(\# u_{start}x_1)$ is a prefix of $v_1^\omega$. Since $w_1$ does not contain $\#$ while $v_1$ does, and by Condition~(1), $h(x_1)w_1$ is also a prefix of $v_1^\omega$, we must have $|w_1| < |v_1|$ and so, the suffix $y_1$ of length $|v_1| - |w_1|$ of $h(x_1)$ is well defined and we have $h(\# u_{start} x_1) = h(x_1 w_1) y_1$. Moreover, since $|y_i| + |w_i| = |v_i \leq |h(x_i)|$, $w_1 y_1$ is also a suffix of $h(x_1)$.

Proceeding by induction, let $1 \leq i < n$ and suppose that \[h(\# u_{start} x_1 \cdots x_i) = h(x_1 w_1 \cdots x_i w_i) y_i.\] By Conditions~(1) and~(2), $ y_{i} {w^\prime_i} h(x_{i+1}) = v_{i+1} h(x_{i+1})$ is a prefix of $v_{i+1}^\omega$. Since $w_i$ does not contain $\#$ while $v_{i+1}$ does, and since by Condition~(1), $h(x_{i+1}) w_{i+1}$ is also a prefix of $v_{i+1}^\omega$, we must have $|w_{i+1}| < |v_{i+1}|$, and so the suffix $y_{i+1}$ of length $|v_{i+1}|-|w_{i+1}|$ of $h(x_{i+1})$ is well defined, and we have $h(x_{i+1} {w_{i+1}}) y_{i+1} = y_i {w^\prime_i} h(x_{i+1})$. Consequently, recalling that
 $ h(x_1 w_1 \cdots x_i w_i) y_i = h(\# u_{start} x_1 \cdots x_i)$,
\begin{align*}
h(x_1 \; \cdots \; x_i \; w_i \; x_{i+1} \; w_{i+1})\; y_{i+1} &= h(x_1 \; \cdots \; x_i  \; w_i) \; y_i \; {w^\prime_i} \; h(x_{i+1})\\
& = h(\#u_{start} \; x_1 \cdots \; x_i \; {w^\prime_i} \; x_{i+1}).
\end{align*}
Moreover, since $|y_i| + |w_i| = |v_i| \leq |h(x_i)|$, it follows that $w_1 y_1$ is also a suffix of $h(x_1)$.
Hence, for all $i, 1 \leq i \leq n$, there exists $z_i (=y_i)$ such that $w_i z_i$ is a  suffix of $h(x_i)$ such that $h(x_1 w_1 \cdots  x_{i} w_i) z_{i} = h(\# u_{start} x_1 \cdots {w^\prime_{i-1}} x_{i})$. 
It remains to show that $h$ is a solution to $\alpha_\mu = \beta_\mu$. This follows from the fact that, as we have just seen, $h(x_1 w_1 \cdots  x_{n} w_n) y_{n} = h(\# u_{start} x_1 \cdots {w^\prime_{n-1}} x_{n})$ and furthermore, by Condition~(3), $y_n {w^\prime_n} h(x_{n+1}) = h(x_{n+1}) \# u_{end}$. Thus
\begin{align*}
h(x_1 w_1 \cdots  x_{n} w_n x_{n+1} \# u_{end}) &= h(x_1 w_1 \cdots  x_{n} w_n) y_{n} {w^\prime_n} h(x_{n+1})\\
&= h(\# u_{start} x_1 \cdots {w^\prime_{n-1}} x_{n} {w^\prime_n} x_{n+1})
\end{align*}
so $h(\alpha_\mu) = h(\beta_\mu)$, and $h$ is an overlapping solution to the equation.

Now suppose that $h$ is an overlapping solution to $\alpha_\mu = \beta_\mu$. Then there exists a proper suffix $z_1$ of $h(x_1)$ such that $h(x_1 w_1) z_1 = h(\# u_{start} x_1)$. Since $h(x_1) \geq |w_1z_1|$, this implies that $h(x_1)$ has a prefix $v_1 = \# u_{start}$ and period $|\#u_{start}|$. This implies that $\# u_{start} h(x_1)$ -- and thus also $h(x_1) w_1$ -- are prefixes of $v_1^\omega$, so Conditions~(1) and~(2) are satisfied for $i = 1$. Moreover, we note that $|z_1| = |v_1|- |w_1| = |y_1|$ so $z_1 = y_1$.

Proceeding by induction, suppose that Conditions~(1) and~(2) are satisfied for $i\leq j$, and furthermore, that $h(x_1 w_1 \cdots  x_{j} w_j) y_{j} = h(\# u_{start} x_1 \cdots x_{j})$. Then, since $h$ is an overlapping solution, there exists a proper suffix $z_{j+1}$ of $h(x_{j+1})$ such that
\begin{align*}
h(x_1 w_1 \cdots  x_{j} w_j x_{j+1} w_{j+1}) z_{j+1} &= h(\# u_{start} x_1 \cdots x_j {w^\prime_j} x_{j+1})\\
&= h(x_1 w_1 \cdots  x_{j} w_j) y_{j} h({w^\prime_{j}} x_{j+1}),
\end{align*}
so $ y_j {w^\prime_j} h(x_{j+1}) = h(x_{j+1}w_{j+1}) z_{j+1}$. Since $|h(x_{j+1})| \geq |w_{j+1}z_{j+1}|$, this implies that $h(x_{j+1})$ has prefix $v_{j+1} = y_j {w^\prime_j}$ and period $|y_j{w^\prime_j}|$. This implies that $v_{j+1} h(x_{j+1})$ -- and thus also $h(x_{j+1}w_{j+1})$ -- are prefixes of $v_{j+1}^\omega$, so Conditions~(1) and~(2) are satisfied for $i = j+1$. Moreover, $|z_{j+1}| = |y_{j} {w^\prime_j} | - |w_j| = |y_{j+1}|$, so $z_{j+1} = y_{j+1}$, and our induction condition is also satisfied for $i = j+1$.

Thus Conditions~(1) and~(2) are satisfied for all $i, 1\leq i \leq n$ and, additionally, we have $h(x_1 w_1 \cdots  x_{n} w_n) y_{n} = h(\# u_{start} x_1 \cdots {w^\prime_{n-1}} x_{n})$. Since $h$ is a solution to $\alpha_\mu = \beta_\mu$, we also have:
\begin{align*}
h(x_1 w_1 \cdots  x_{n} w_n x_{n+1} \# u_{end}) &= h(\# u_{start} x_1 \cdots {w^\prime_{n-1}} x_{n} {w^\prime_n} x_{n+1})\\
&= h(x_1  w_1 \cdots  x_{n} w_n) y_n {w^\prime_n} h(x_{n+1})
\end{align*}
so $h(x_{n+1}) \# u_{end} = y_n {w^\prime_n} h(x_{n+1})$ and $h$ also satisfies Condition~(3).
\end{proof}

\noindent {\bf Proof of Theorem~\ref{prop:1repvar}}

\noindent We need the following additional preliminaries. Two words are prefix (resp. suffix)-compatible if one is a prefix (resp. suffix) of the other. A primitive word is one which is not a repetition of a shorter word. Recall that for a word $u$, $u^\omega$ is the infinite word obtained by repeating $x$. 
We also need the following folklore lemmas. Note that a primitive word is one which is not a repetition of a strictly shorter word (i.e. $u$ is primitive if $u = v^n$ implies $n=1$).
\begin{lemma}[Fine and Wilf]\label{lem:fineandwilf}
If $u$,$v$ are primitive words and $u^\omega$ and $v^\omega$ have a common prefix of length at least $|u| + |v| - gcd(|u|,|v|)$,  then $u = v$.
\end{lemma}

Note that as a consequence of the lemma, if, for primitive words $u$ and $v$, several consecutive $u$s overlap with several consecutive $v$s, $u$ and $v$ are conjugate.

\begin{lemma}\label{lem:conjugacyequation}
Suppose that $x,y,z \in \Sigma^*$ such that $xy = yz$. Then there exist $u,v \in \Sigma^*$ and $p,q \in\mathbb{N}_0$ such that $x = (uv)^p$, $y = (uv)^q u$ and $z =(vu)^q$ where $uv$ is primitive. 
\end{lemma}

We also have the following technical lemma.

\begin{lemma}\label{lem:systemofequations}
Let $x,y$ be variables and let $A_1,A_2,\ldots A_k, B_1,B_2,\ldots, B_k \in \Sigma^*$. Let $\Phi$ be the the system of equations
\begin{align*}
A_1 x &= y B_1\\
A_2 A_1 x &= y B_1 B_2\\
&\;\;\vdots\\
A_k \ldots A_1 x &= y B_1 B_2 \ldots B_k.
\end{align*}
A substitution $h: (\{x,y\}\cup\Sigma)^* \to \Sigma^*$ with $|h(y)| > 2|A_kA_{k-1}\ldots A_1|$ is a solution to $\Phi$ if and only if there exist $u,v \in \Sigma^*$ with $|uv| \leq |A_kA_{k-1}\ldots A_1|$ and $p, q_2, q_3, \ldots q_k \in \mathbb{N}_0$ such that $|(uv)^pu| > 2|A_kA_{k-1}\ldots A_2|$ and:
\begin{enumerate}
\item{}$h(x) = (uv)^p u B_1$ and $h(y) = A_1 (uv)^p u$, and
\item{}$A_i \ldots A_2  = (uv)^{q_i}$, $B_2 \ldots B_i = (vu)^{q_i}$ for each $i$, $2 \leq i \leq k$.
\end{enumerate}
\end{lemma}

\begin{proof}
Suppose $h$ is a substitution with $|h(y)| > 2|A_kA_{k-1}\ldots A_1|$. Since $|h(y)| \geq |A_1|$, if $h$ solves the first equation, then there exists $w\in \Sigma^*$ such that $h(x) = wB_1$ and $h(y) = A_1 w$. Note that $w > 2|A_i\ldots A_2|$. Moreover, $h$ also satisfies the whole system if and only if:
\begin{align*}
A_2 w &= w B_2\\
&\;\;\vdots\\
A_k \ldots A_2 w &= w B_2 \ldots B_k
\end{align*}
By Lemma~\ref{lem:conjugacyequation}, $w$ satisfies $A_i \ldots A_2 w = w B_2 \ldots B_i$ if and only if there exist $u_i,v_i \in \Sigma^*$ and $q_i,p_i \in \mathbb{N}_0$ such that $u_iv_i$ is primitive, $w = (u_iv_i)^{p_i} u_i$, $A_i \ldots A_2 = (u_iv_i)^{q_i}$ and $B_2\ldots B_i = (v_iu_i)^{q_i}$.

Now, if $h$ is a solution,  since $|u_iv_i| \leq |A_i\ldots A_2|$ for each $i$, and since $w > 2|A_i\ldots A_2|$, we must have that $p_i \geq 2$. Furthermore, we have
\[
w = (u_2v_2)^{p_2}u_2 = (u_3v_3)^{p_3}u_3 = \cdots = (u_kv_k)^{p_k}u_k
\]
and since each $u_iv_i$ is primitive, by Lemma~\ref{lem:fineandwilf}, $u_iv_i = u_jv_j$ for all $i,j$ (and hence that each $p_i = p$ for some fixed value $p$), so the conditions of the Lemma are satisfied. On the other hand, if the conditions of the lemma are satisfied, then it is straightforward to see that $h$ is a valid solution. 
\end{proof}

We are now ready to prove the main statement. 

\begin{proof}
Let $E: \alpha = \beta$ be an equation in $\mathcal{D}$.
If both $\alpha$ and $\beta$ contain at least two variables, then we may refer to~\cite{ManSchNow16}. Hence w.l.o.g.\ we assume that $\var(\alpha) = \{x\}$. For the simplicity of the exposure, we shall only prove completely the case that $\beta$ contains only one variable either side of the repeated variable $x$. The general case is a straightforward adaptation of the proof. Hence our equation has the form
\[ x u_1 x u_2 \cdots x_k x = y v_0 x v_1 \cdots x v_k z,
\mbox{ for some }k \in \mathbb{N}.\]

Firstly, we note that by using the method of filling the positions (cf.~\cite{PlandowskiR98}), we can check whether a solution with specific lengths of $x, y, z$ exists in polynomial time with respect to the sum of the lengths. Hence it is sufficient to show that for a minimal solution, these lengths are bounded by some polynomial of the length of the equation. Moreover, if the length of (the image of) $x$ is bounded by a polynomial, then so is the length of the whole solution word, and hence the images of the variables $y$ and $z$.

Suppose that $g$ is a minimal solution to the equation and in particular, assume that $|g(x)| > |\alpha\beta|$ (otherwise we are done). We may also assume that $|g(y)| < |g(x)|$ or $|g(z)| < |g(x)|$, since $ (n+1)|g(x)| + |u_1u_2\ldots u_k| = n|g(x)| + |g(z)| + |g(y)| + |v_0v_1\ldots v_k|$, so if $|g(z)|,|g(y)| > |g(x)|$, we have that $|g(x)|<|E|$. W.l.o.g.\ let $h(g)| < |g(x)|$.

Now suppose that there exist $i \leq k,j < k$ such that
\[ |g(x u_1 \ldots x u_i)| = |g(y v_0 x \ldots v_{j-1} x)| + \ell
\]
for some $\ell, 0 \leq \ell < |v_j|$ (i.e., so that a suffix of $u_i$ `overlaps' with a prefix of $v_j$). Then $g(x)$ has prefix $v_j[\ell+1\ldots|v_j|]$ and period $|v_j|-\ell$. Thus there exist $s,t\in \Sigma^*$ such that $st = v_j[\ell+1\ldots|v_j|]$ and $g(x) = st^p s$ for some $p \in \mathbb{N}$. It is straightforward to see that when $p$ is ``large'' (e.g., greater than $|E|$) that the morphism $g^\prime$ given by $g^\prime(x) = st^{p-1}s$, $g^\prime(z) = g(z)$ and $g^\prime(y) = g(y)$ is also a solution.
A symmetric argument holds for the case that
\[ |g(x u_1 \ldots  u_{i-1} x)| +\ell = |g(y v_0 x \ldots  x v_j)| 
\]
for some $\ell, 0 \leq \ell < |u_i|$ (i.e., so that a prefix of $u_i$ `overlaps' with a suffix of $v_j$). Hence the length of any minimal solution is bounded by a polynomial of $E$ whenever two of the terminal/constant parts of the equation overlap in the solution. Therefore, for the remainder of the proof, we may assume that
\[ |g(y v_0) | < |g(x)| \leq |g(x u_1)| < |g(y v_0 x)| \leq |g(y v_0 x v_1)| < |g(x u_1 x)| \leq \ldots < |g(y v_0 x \ldots x)|\]
(in other words, that all the occurrences of $g(x)$ ``overlap''), in which case the valid solutions are characterised by solutions $h$ to the following system of equations $\Psi$ (in which $h(x_0)$ correlates to $g(y)$, $h(x_{2k+1})$ correlates to $h(z)$, and the other variables $h(x_i)$ correlate to the overlapping parts of $g(x)$).

\begin{align*}
\Psi: \;\;x &= x_0 \; v_0 \; x_1\\ 
& = x_1 \; u_1 \; x_2\\
&= x_2 \; v_1 \; x_3\\
&= x_3 \; u_2 \; x_4\\
&\qquad\quad\! \vdots\\
&= x_{2k} \; v_k \; x_{2k+1}.\\
\end{align*}

In fact, we observe that $\Psi$ is equivalent to the (union of the) systems $\Psi_1,\Psi_2,\Psi_3,\Psi_4$ given as follows, and note that any for any solution $h_1$ to $\Psi_1 \cup \Psi_2 \cup \Psi_3 \cup \Psi_4$, there exists an equivalent solution $h_2$ to $\Psi$ with $h_1(x_i) = h_2(x_i)$ for $1\leq i \leq 2k+1$ and vice-versa. Thus it is sufficient to show that the minimal solution to $\Psi_1 \cup \Psi_2 \cup \Psi_3 \cup \Psi_4$ is sufficiently short.
\begin{align*}
 \Psi_1: &\textcolor{white}{\;=\;}x_0 \; v_0 \; x_1 & \Psi_2 : \textcolor{white}{=\;} y_1 \; u_1 \; y_2\\ 
&= x_2 \; v_1 \; x_3 & = y_3 \; u_2 \; y_4\\
&\qquad\quad\! \vdots & \vdots \qquad\\
&= x_{2k} \; v_k \; x_{2k+1}, & = y_{2k-1} \; u_{k} \;y_{2k},\\
\end{align*}

\begin{align*}
\Psi_3: \;\;x_1 &= y_1 & \Psi_4:\;\;y_1 u_1 y_2 = x_2 v_1 x_3.\\
x_2 &= y_2\\
&\;\;\vdots\\
x_{2k} &= y_{2k},
\end{align*}

We need the following claim bounding the length-difference between two $h(x_i)$s for any solution $h$ to the above system with indicies of the same parity.

\begin{claim}\label{claim:lengths}
Let $h$ be a solution to $\Psi$ (or, equivalently, $\Psi_1 \cup \Psi_2 \cup \Psi_3 \cup \Psi_4$). Let $i,j \in \{0, 2,\ldots,2k\}$. Then \[||h(x_i)| - |h(x_j)|| \leq |E|.\]
The same statement holds when $i,j \in \{1,3,\ldots, 2k+1\}$.
\end{claim}

\begin{proof}
Let $n = |h(x_0)|$, and let $ m = |h(x_1)|$. Then we have that $|h(x)| = |h(x_i) w_i h(x_{i+1})|$ for $0 \leq i \leq 2k$, where $w_i = v_{\frac{i}{2}}$ if $i$ is even and $w_i = u_{\frac{i+1}{2}}$ if $i$ is odd. Hence,
\begin{align*}
|h(x_{i+1})| &= |h(x)| - |w_i| - |h(x_i)|\\
& =  |h(x_{i-1})| + |w_{i-1}| + |h(x_{i})|  - |w_i| - |h(x_i)|  \\
&=  |h(x_{i-1})| + |w_{i-1}|  - |w_i|.
\end{align*}
Thus, in general, if $i$ is even, then $|h(x_i)| = n + \sum\limits_{j \text{ even}, j\leq i}|w_{j-2}| -|w_{j-1}|$ and if $i$ is odd, then $|h(x_i)| = m + \sum\limits_{j \text{ odd}, j\leq i}|w_{j-2}| -|w_{j-1}|$. The statement of the claim follows.
\end{proof}

Now, suppose $h$ is a substitution, and let $\tilde x $ be the longest common prefix of $h(x_0), h(x_2),$ $\ldots, h(x_{2k})$ and let ${\tilde x}^\prime$ be the longest common suffix of $h(x_1),h(x_3),\ldots, h(x_{2k+1})$. Clearly, $h$ is a solution to $\Psi_1$ if and only if there exist $A_0,A_2,\ldots A_{k}$, $B_1,B_2,\ldots B_{k+1} \in \Sigma^*$ such that $h(x_i) = \tilde x A_{\frac{i}{2}}$ if $i$ is even and $h(x_i) = B_{\frac{i+1}{2}} {\tilde x}^\prime$ if $i$ is odd, and such that 
$A_{i-1} v_{i-1} B_{i} = A_{i} v_i B_{i+1}$
for all $i,1\leq i \leq k$. 
Moreover, it follows from Claim~\ref{claim:lengths} that each of the lengths $|A_i|$, $|B_i|$ is bounded by $|E|$.

Similarly, let $\tilde y $ be the longest common prefix of $h(x_1),h(x_3),\ldots, h(x_{2k-1})$ and let ${\tilde y}^\prime$ be the longest common suffix of $h(x_2),h(x_4),\ldots, h(x_{2k})$. Then $h$ is a solution to $\Psi_2$, if and only if there exist $C_1,\ldots C_{k}$, $D_1,D_2,\ldots D_{k}$ such that $h(y_i) = \tilde y C_{\frac{i+1}{2}}$ if $i$ is odd and $h(y_i) = D_{\frac{i}{2}} {\tilde y}^\prime$ if $i$ is even, and such that 
$C_{i} u_i D_{i} = C_{i+1} u_{i+1} D_{i+1}$ for all $i,1\leq i < k$ where by Claim~\ref{claim:lengths}, $|C_i|$ and $|D_i|$ are bounded by $|E|$.

Suppose that $h$ is a solution to $\Psi_1$ and $\Psi_2$ and hence that it satisfies the conditions above. It is a straightforward observation that $h$ is \emph{also} a solution to $\Psi_3$ if and only if the following systems of equations are satisfied in addition:
\[ \Phi_1: \;\; \tilde y C_i = B_i {\tilde x}^\prime, \;\; 1\leq i \leq k \]
\[ \Phi_2: \;\; \tilde x A_i = D_i {\tilde y}^\prime, \;\; 1\leq i \leq k.\]

Moreover, we can infer from the equations in $\Psi_1$ and $\Psi_2$ that the $D_i$ factors are pairwise suffix compatible, and the $A_i$ factors are pairwise prefix compatible. Likewise, $B_i$ factors are pairwise suffix compatible while the $C_i$ factors are pairwise prefix compatible. Hence there exist each of the two systems above can be written as a system of the form described by Lemma~\ref{lem:systemofequations}. We consider 4 cases based on whether $\tilde x$ and/or$\tilde y$ are long. 

Our first case is that $|\tilde y| \leq 2\max(|B_i|)$ and $|\tilde x| \leq 2\max(|D_i|)$. Then  if $h$ is a solution to $\Phi_1$ and $\Phi_2$, the lengths $|{\tilde x}^\prime|, |{\tilde y}^\prime|$ are similarly bounded. Hence $|h(y_1 \alpha_1 y_2)| = |\tilde y C_1 \alpha_1 D_1 {\tilde y}^\prime |$ is in $O(|E|)$, and consequently, the minimal solution $g$ to $E$ has length bounded by a polynomial of $|E|$ and we are done.

Our second case is that $|\tilde y| \leq 2\max(|B_i|)$ and $|\tilde x| > 2\max(|D_i|)$. Note that this corresponds to the case that $|h(y)| > 2|A_kA_{k-1}\ldots A_1|$ when translating $\Phi_2$ into the terms of Lemma~\ref{lem:systemofequations}. If $h$ satisfies $\Phi_1$, we have that $|{\tilde x}^\prime$ is in $O(|E|)$. By Lemma~\ref{lem:systemofequations}, if $h$ is also a solution to $\Phi_2$ if and only if $u,v \in \Sigma^*$ and $i,j,k \in \mathbb{N}_0$ such that  $\tilde y = B_i (uv)^k u $ and ${\tilde x}^\prime = (u v)^{k} u C_{j}$ where $|uv|$ is bounded by $|E|$, and $uv$ is primitive. Thus $h$ also satisfies $\Psi_4$ if and only if:
\begin{align*}
h(y_1 \alpha_1 y_2) &= h(x_2 \beta_1 x_3)\\
\implies  \tilde y C_1 \alpha_1 D_2 {\tilde y}^\prime &= \tilde x A_2 \beta_1 B_3 {\tilde x}^\prime\\
\implies  B_i (uv)^k u C_1 \alpha_1 D_2 {\tilde y}^\prime &= \tilde x A_2 \beta_1 B_3 (u v)^{k} u C_{j}
\end{align*}
Since $uv$, is primitive (and therefore does not overlap with itself in a non-trivial way), it is clear that if a solution $h$ exists satisfying the above equation, then such a solution exists for (polynomially) small $k$ (as soon as $k$ is large enough that some of the $uv$ factors overlap, we also have a solution for $k-1$ so $h$ is not minimal which contradicts our assumption). Consequently, Since all the factors have length bounded by $|E|$, this is sufficient to show that $h$, and thus any minimal solution to $E$, is has length at most polynomial in $|E|$.

The case that $|\tilde y| > 2\max(|B_i|)$ and $|\tilde x| \leq 2\max(|D_i|)$ may be treated identically. Finally, suppose that both $|\tilde y| > 2\max(|B_i|)$ and $|\tilde x| > 2\max(|D_i|)$ (note that these correspond to the cases that $|h(y)| > 2|A_kA_{k-1}\ldots A_1|$ when translating into the terms of Lemma~\ref{lem:systemofequations}). Then by Lemma~\ref{lem:systemofequations},   $h$ satisfies $\Phi_1$ and $\Phi_2$, if and only if there exist $u,v,u^\prime,v^\prime$, $i$,$j$, $i^\prime,j^\prime$, $k, k^\prime$ such that $\tilde y = B_i (uv)^k u $, ${\tilde y}^\prime = (u^\prime v^\prime)^{k^\prime} u^\prime A_{i^\prime}$ and $\tilde x = D_j (u^\prime v^\prime)^{k^\prime} u^\prime $, ${\tilde x}^\prime = (u v)^{k} u C_{j^\prime}$ where $|uv|, |u^\prime v^\prime|$ are bounded by $|E|$, and $uv$, $u^\prime v^\prime$ are primitive. Thus $h$ also satisfies $\Psi_4$ if and only if:
\begin{align*}
h(y_1 \alpha_1 y_2) &= h(x_2 \beta_1 x_3)\\
\implies  \tilde y C_1 \alpha_1 D_2 {\tilde y}^\prime &= \tilde x A_2 \beta_1 B_3 {\tilde x}^\prime\\
\implies B_i (uv)^k u C_1 \alpha_1 D_2 (u^\prime v^\prime)^{k^\prime} u^\prime A_{i^\prime} & = D_j (u^\prime v^\prime)^{k^\prime} u^\prime A_2 \beta_1 B_3 (u v)^{k} u C_{j^\prime}.
\end{align*}
As before, since $uv$, $u^\prime v^\prime$ are primitive, it is reasonably straightforward using standard arguments from combinatorics on words that if such a solution $h$ exists satisfying the above equation, then a solution exists for small $k$ and $k^\prime$, since many overlapping $uv$s or $u^\prime v^\prime$s again means that some repetitions may be removed and thus the solution is not minimal. Again all factors have length bounded by $|E|$, so $h$, and thus any minimal solution to $E$, is has length at most polynomial in $|E|$ and the statement of the theorem follows.
\end{proof}

\section*{$\npclass$-upper bounds for equations with regular constraints}

For a word equation $\alpha = \beta$ and an $x \in \var(\alpha \beta)$, a \emph{regular constraint} (\emph{for $x$}) is a regular language $L_x$. A solution $h$ for $\alpha = \beta$ \emph{satisfies} the regular constraint $L_x$ if $h(x) \in L_x$. The satisfiability problem for word equations with regular constraints is to decide on whether an equation $\alpha = \beta$ with regular constraints $L_{x}$, $x \in \var(\alpha \beta)$, given as an NFA, has a solution that satisfies all regular constraints. 

Let us first note that the satisfiability of regular-ordered equations with (general) regular constraints is $\pspaceclass$-complete follows from \cite{ManSchNow16}. In the following we consider the case of regular-ordered equations with regular constraints, when the regular constraints are regular languages that are all accepted by nondeterministic finite automata (NFA) with at most $c$ states, where $c$ is a constant (also called constant regular constraints). 

\begin{theorem}The satisfiability problem for regular-ordered equations with constant regular constraints is in NP.
\end{theorem}
\begin{proof}
We analyse regular-ordered equations $E: \alpha = \beta$ with regular constraints, such that for all $x\in \var(\alpha\beta)$ the language $L_x$ is accepted by an NFA with at most $c$ states (where $c$ is a constant). Let $n=|\alpha\beta|$. A trivial remark is that if the language $L_x$ is accepted by an NFA with at most $c$ states then it is accepted by a DFA (denoted $A_x$ in the following) with at most $2^c$ states, which is still a constant. 
For simplicity, let $C=2^{c}+3$. Let $K$ be the number of DFAs with input alphabet $\Sigma$ and at most $C$ states; it is immediate that $K$ is constant (although exponential in $C$, so doubly exponential in $c$). Also, let $A_1,\ldots, A_K$ be an enumeration of the DFAs with at most $C$ states. 

In the following we show that the minimal solution to a regular-ordered equation $E: \alpha = \beta$ with regular constraints as above has length $O(n^4)$, with the constant hidden by the $O$-notation being exponential in $K$. 

We will use in the following the same notations as in Proposition \ref{prop:linearsolutions}. Let $h$ be a minimal solution to $E$ and let $H=h(\alpha)=h(\beta)$. 

Firstly, we note that in the minimal solution $h$ to $E$, unlike the case of equations without regular constraints, it is not necessary that every position of $h$ occurs somewhere in one of the sequences $S_p$ that start or end with a terminal symbol of the equation. Now, because of the regular constraints, we might need some "hidden" factors inside the images of the variables, whose symbols do not belong to any sequence starting or ending with a terminal symbol; such factors ensure that the variable-image to which they belong satisfies its regular constraint. It is straightforward to note that the sequences that contain symbols of these hidden factors start with a single occurring variable and end with a single occurring variable. Therefore, they will belong to so-called {\em invisible sequences}. We define the invisible sequences as follows. 

For each position $p$ associated with a variable occurring only once in $\alpha\beta$, we construct a sequence $S_p = p_1,p_2,\ldots$ (called invisible sequence) such that
\begin{itemize}
\item{}$p_1 = p$ and $p_2$ is the (unique) position corresponding with $p_1$, and 
\item{}for $i \geq 2$, if $p_i = (x,z,d)$ such that $x$ occurs only once in $\alpha\beta$, then the sequence terminates, and
\item{}for $i \geq 2$, if $p_i = (x,z,d)$,  such that $x$ is a variable occurring twice, then $p_{i+1}$ is the position corresponding to the (unique) position $(x,z^\prime,d)$ with $z^\prime \not= z$ (i.e., the `other' occurrence of the $i^{th}$ letter in $h(x)$). 
\end{itemize} 
Moreover, we can talk about {\em similarity classes} of invisible sequences: all similar invisible sequences are grouped in the same similarity class. For simplicity, let us assume that invisible sequences always start with the leftmost of the single occurring variables between which it extends (i.e., the variable whose image in the minimal solution $h$ has its first symbol closer to beginning of $H$).

Our proof is based on three claims regarding the structure of a minimal solution $h$ of $E$:
\begin{itemize}
\item[1.] We first show that each sequence (regular or invisible) contains $O(n)$ elements (where the constant hidden by the $O$-notation is linear in $C$).
\item[2.] The number of invisible sequences similar to a given sequence is $O(1)$ (where the constant hidden by the $O$-notation is proportional to $(C+1)^{CK}$). 
\item[3.] The number of similarity classes of invisible sequence is $O(n^3)$ (where the constant hidden by the $O$-notation is linear in $C$).  
\end{itemize}

To prove Claim 1 from above we use the same general strategy as in Proposition \ref{prop:linearsolutions}. 

Let $S_p$ be any sequence (regular or invisible). 
Due to the structure of the equations, $S_p$ cannot contain subsequences $\ldots, (x,z_1,d_1), \ldots, (y,z_2,d_2), \ldots, (x,z_3,d_3), \ldots$. Moreover, it is important to note that if $S_p$ contains a subsequence $(x,z_1,d_1),\ldots, (x,z_i,d_{i})$ then $d_1<d_2<\ldots <d_i$.

Let us assume that $S_p$ contains a subsequence $(x,z_1,d_1),\ldots, (x,z_C,d_{C})$. Let $q_i$ be the state in which the automaton $A_x$ enters after reading the word $h(x)[1..d_i-1]$. It is immediate that by the form of the equation $z_i=z_j$ for all $i,j$. By Lemma \ref{lem:squares1} (and its proof) we get that $h(x)[d_i..d_{i+1}-1]=h(x)[d_{i+1}..d_{i}-1]=v$ for all $i< C$; let $p=|h(x)[d_i..d_{i+1}-1]|=d_{i+1}-d_i$. 

As $C$ is strictly greater than the number of states of $A_x$ we get that there exists $i^\prime$ and $i^{\prime\prime}$, with $1<i^\prime<i^{\prime\prime}<C$ such that $q_{i^\prime}=q_{i^{\prime\prime}}$. Also, we have $h(x)[d_{i^\prime}..d_{i^{\prime\prime}}-1]$ corresponds to the factor $H[j..j^\prime]$ of $H=h(\alpha)$ but also to the factor $H[j+p..j^\prime+p]$ of $H=h(\beta)$; in both cases, these factors are both succeeded and followed by another $v$, which occur completely inside $h(x)$. By a reasoning similar to Lemma \ref{lem:squares1} we immediately get that we can obtain a shorter solution of our equation by removing the factor $h(x)[d_{i^\prime}..d_{i^{\prime\prime}}-1]$ from the image of $x$, noting that since $q_{i^\prime}=q_{i^{\prime\prime}}$, $h(x)$ still satisfies the regular constraint.
\begin{figure}[ht]
\vspace*{-0.3cm}
	\centering
  \includegraphics[width=\textwidth ]{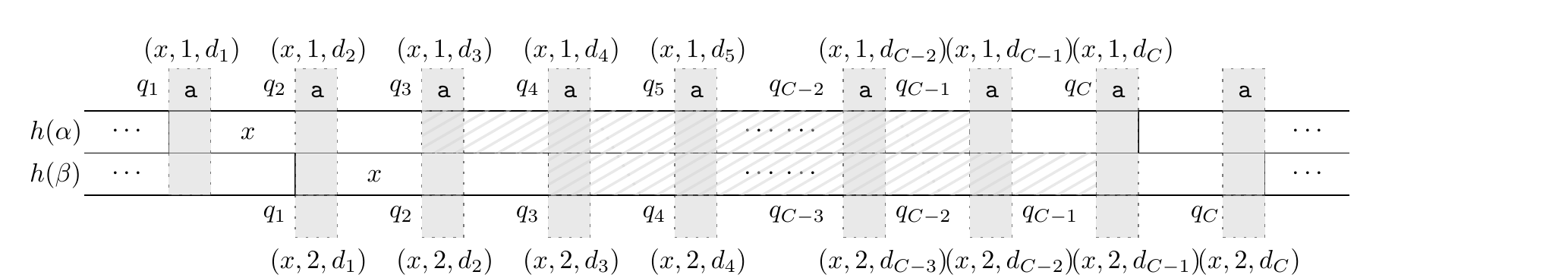}
	\caption{A subsequence $\ldots, (x,1,d_1), (x,1,d_2),(x,1,d_3),\ldots , (x,1,d_{C-1}), (x,1,d_{C}), \ldots$ in a solution of the equation, defined by $h$. If $q_3=q_{C-1}$, as $h(x)[d_2..d_3-1]=h(x)[d_{C-1}..d_C-1]$, then $h(x)[d_3..d_{C-1}-1]$ (shaded in the figure) can be removed from the factors $h(x)$ occurring on both sides to obtain a shorter solution of the equation.}	
	\label{fig4}
\vspace*{-0.3cm}
\end{figure}

Therefore, the number of times a sequence $S_p$ (which is part of the minimal solution of $E$) can contain a triple with $x$ on the first position is strictly smaller than $C$. In conclusion, the length of each sequence $S_p$ in the minimal solution is upper bounded by $Cn$. This concludes the proof of Claim 1. The next claims help us upper bound the total number of sequences.

Let us now move on and show Claim 2. Let $M=(C+1)^{CK}+1$. Let us assume that there exist $S_1, S_2,\ldots , S_M$ similar invisible sequences. Let us assume that $S_i$ starts with $(x,1,d_i)$ for all $i\leq M$, with $d_1<d_2<\ldots<d_M$. It is not hard to note that if we consider the element $(y,z,d^\prime_i)$ occurring in each sequence $S_i$ with $1\leq i\leq M$ on position $\ell$ (same for all sequences), then $h(y)[d^\prime_i..d^\prime_j]=h(x)[d_i..d_j]$ for all $1\leq i<j\leq M$. As the number of regular languages that may define the regular constraints used in our equation is constant, it follows that it may be the case that the regular constraints associated to different variables traversed by our sequences $S_i$ are actually the same. Now, let $(y_1,z^i_1,e^i_1),\ldots,(y_t,z^i_t,e^i_t)$ be all the elements of the sequence $S_i$ whose variables $y_j$ is subject to the regular constraint accepted by $A_1$, in the order they appear in this sequence. It is worth noting that the order of these variables and their relative position is exactly the same in all sequences, because these sequences are similar. 
Let now $q^1_j$ be the state in which $A_1$ enters after reading $h(y_j)[1..e^1_j-1]$; this state determines uniquely the state $q^t_j$ in which $A_1$ enters after reading $h(y_j)[1..e^t_j-1]$ for $2\leq t\leq M$. There are only $C$ possibilities to choose the beginning state $q^1_j$; so, if we consider $s_1,\ldots,s_C$ an enumeration of the states of $A_1$, we will consider, in order, the cases of $q^1_j$ being each of these states. For $i\leq C$, let $J_{s_i}=\{j\mid q^1_j=s_i\}$; clearly $q^t_{j_1}=q^t_{j_2}$ for all $1\leq t\leq M$ and $j_1,j_2\in J_{s_1}$. Now, for some $j\in J_{s_1}$ we have that the states $q^t_j$ with $1\leq t\leq M$ can take at most $C$ different values, so there exists a subset $M^{(0)}$ of $\{1,\ldots,M\}$ with at least $M/C$ elements such that all states $q^t_{j}$ with $t\in M^{(0)}$ are all equal to a state $s^{(0)}$ of $A_1$.  Further, we consider $j\in J_{s_2}$ and have again that the states $q^t_j$  with $ t\in M^{(0)}$ (which are, again, identical for all values $j\in J_{s_2}$) can take at most $C$ different values, so at least $|M^{(0)}|/C$ of them should be identical. Then, we take the subset $M^{(1)}\subseteq M^{(0)}$ with least $|M^{(0)}|/C$ elements such that all states $q^t_{j}$ with $t\in M^{(1)}$ are all equal to a state $s^{(1)}$ of $A_1$. We keep repeating this procedure until we finished considering the case $q^1_j=s_C$ and obtained a set $M^{(C)}$. It follows immediately that there is a subset $M_1=\{S_i\mid i\in M^{(C)}\}$ of $\{S_1, \ldots, S_M\}$ of size at least $M/C^C\geq (C+1)^{C(K-1)}+1$ such that, for each $j$, we have that $A_1$ enters the same state after reading $h(y_j)[1..e^t_j-1]$ for all $t\in M_1$. We then repeat the same reasoning with $M_1$ in the role of $S_1,\ldots,S_M$ and $A_2$ instead of $A_1$ to produce an even smaller set $M_2$, of size at least $|M_1|/C^C\geq (C+1)^{C(K-2)}+1$. Further we repeat this procedure for each $A_i$, with $1\leq i\leq K$, and the set $M_{i-1}$ produced in the previous step. In the end we reach a subset $M_K$ of $\{S_1,\ldots ,S_M\}$, with at least $2$ elements $m_1<m_2$, such that for all elements $(y,z,d_{m_1})$ and $(y,z,d_{m_2})$ occurring on the same position in the sequences $S_{m_1}$ and $S_{m_2}$, respectively, we have that the automaton $A_y$ (accepting the variable $y$) enters after reading $h(y)[1..d_{m_1}-1]$ the same state as the state it enters after reading $h(y)[1..d_{m_2}-1]$. It is not hard to see that, in this case, we can remove from all the images of the variables found on our similar sequences, respectively, the factor corresponding to $h(y)[d_{m_1}..d_{m_2}-1]$, and get a shorter solution to our equation that still fulfils the regular constraints. This is a contradiction to the minimality of $h$, so, in conclusion, we cannot have $M$ similar invisible sequences. This concludes the proof of Claim~2.

We finally show Claim 3. We say that two invisible sequences $S_1$ and $S_2$ split if there exists $\ell$ such that the $i^{th}$ elements of $S_1$ and $S_2$ are similar, for all $i<\ell$, and the $\ell^{th}$ elements of $S_1$ and $S_2$ are not similar. Let us now consider only the invisible sequences starting with a variable $x$ (single occurring). Due to the particular form of the equations, it is clear that if $S_1$ and $S_2$ are two such invisible sequences which are not similar and $S_1$ starts to the left of $S_2$ (w.r.t. the solution word $h(\alpha)$) and no other invisible sequence starting between them, then $S_1$ cannot be similar to any invisible sequence starting on a position of $h(x)$ to the right of the starting position of $S_2$. So, essentially, the similarity classes of invisible sequences starting with $x$ can be bijectively associated to the splits between consecutive invisible sequences. 

So, let us consider two such consecutive sequences $S_1$ and $S_2$, with $S_1$ starting to the left of $S_2$ and no other invisible sequence starting between them. Assume that $S_1$ and $S_2$ are the first to split among all pairs of consecutive invisible sequences starting in $x$; more precisely, assume that they split after $\ell$ elements of the sequence (and they belong to different similarity classes), and all other pairs of invisible sequences split after at least $\ell$ elements. This split occurs because $S_1$ reaches a triple $(y,z,|h(y)|)$ and $S_2$ a triple $(u,z,1)$ with $y$ and $u$ consecutive variables in $\alpha$ or $\beta$ (maybe with terminals between them). For simplicity, we say $S_1$ and $S_2$ are split by $y$ and $u$, and also note that no other pair of consecutive invisible sequences can be split exactly after their first $\ell$ positions by $y$ and $u$, due to the regular ordered form of the solutions. Moreover, it is not hard to see that any two consecutive sequences to the right of $S_2$ can not be split by $y$ and $u$; otherwise there will be a sequence leading from a symbol of $h(u)$ (reached by a sequence on position $\ell$), other than the first one, to the first symbol of $u$ (reached by that sequence when the second split happens), a contradiction. To the left of $S_1$ there still might be pairs of consecutive sequences that are split by $y$ and $u$, but all those sequences must only contain triples with the first component $y$ after the $\ell^{th}$ position, until the split (as they already reached $y$ and the variables cannot alternate in sequences, due to the form of the equations). As each sequence contains at most $Cn$ elements with the same variable, and as there cannot be two distinct pairs of consecutive sequences split by $y$ and $u$ after the same number of elements, we might have at most $Cn$ pairs of consecutive sequences split by $y$ and $u$. Now, as splits can only be caused by variables occurring consecutively in $\alpha$ and $\beta$, we consider each such pair of variables and note that each can split up to $Cn$ consecutive sequences starting in $x$. So, as the number of possible similarity classes is upper bounded by twice the number of splitting points multiplied by $Cn$. We get that the number of classes of similar sequences starting in $x$ is $O(Cn^2)$. The conclusion of Claim 3 follows immediately. 

From our three claims we get immediately that the size of the minimal solution of $E$, which is proportional to the total length of the sequences, is $O(n^4)$ (where the constant is proportional to in $(C+1)^{CK}$). 
It now follows immediately that solving regular-ordered equations with regular constraints accepted by NFAs with at most $c$ states is in $\npclass$. We just have to guess the images of all variables $x\in \var(\alpha\beta)$ and then check whether they are in the respective languages $L_x$ and also whether they satisfy the input equation. \end{proof}

\begin{theorem}The satisfiability problem for regular-ordered equations whose sides contain exactly the same variables, with (unrestricted) regular constraints, is in $\npclass$.
\end{theorem}
\begin{proof}
We will use the same notations as in the previous proof. We analyse regular-ordered equations $E: \alpha = \beta$ with regular constraints, such that $\var(\alpha)=\var(\beta)$ and for all $x\in \var(\alpha\beta)$ the language $L_x$ is accepted by an NFA with at most $m$ states (here $m$ is not a constant anymore), has length $O(n)$. Let $n=|\alpha\beta|$. 

In the following we show that the minimal solution $h$ to a regular-ordered equation $E: \alpha = \beta$ with regular constraints as above has length polynomial in $n$. Let $H=h(\alpha)=h(\beta)$.
Due to the particular form of these equations, there are no single occurring variables. So, when we analyse the sequences $S_p$ of equivalent positions defined by the minimal solution $h$, each of them starts and ends with a terminal symbol (so there are at most $n$ sequences).

Essentially, for each variable $x\in \var(\alpha)$, $h(x)$ has two (not necessarily distinct) occurrences in $H$, induced by the occurrence of $x$ in each side of the equation. These occurrences can either be overlapping or non-overlapping. In the first case, the overlap of the two occurrences of $h(x)$ will lead to sequences that contain subsequences $(x,z_1,d_1),\ldots, (x,z_i,d_{i})$ for some $i>1$. In the second case, there will be in each sequence at most one triple that contains the variable $x$. Moreover, in this case $|h(x)|$ is at most equal to the difference between the length of the string occurring in $H$ before the rightmost occurrence of $h(x)$ and the length of the string occurring in $H$ before the leftmost occurrence of $h(x)$; as these two strings contain exactly the same images of variables, this difference is upper bounded by the difference between the total length of the two sides of the equations, so smaller than $n$. In conclusion, variables whose images are non-overlapping in $H$ have length at most $n$. 

With a proof that follows exactly the lines of the proof of Claim 1 from the previous proof, one can show that if a sequence contains a subsequence $(x,z_1,d_1),\ldots, (x,z_i,d_{i})$ for some $i>1$ then $i$ is upper bounded by $m+2$ (in the respective proof it is enough to use any accepting computation for $h(x)$, not necessarily a deterministic one). This also leads to an upper bound of $(m+2)n$ for the length of any sequence (again the occurrences of different variables cannot be interleaved in a sequence). 

Adding these up, we get that $|H|<(m+2)n^2$, so the length of the image of each variable in the minimal solution of $E$ is polynomial. It follows immediately that our statement holds.
\end{proof}

\end{document}